\documentclass[11pt]{article}
\usepackage{authblk}
\title{Data Fusion with Distributional Equivalence Test-then-pool}
\author[1]{Linying Yang\thanks{Corresponding email: linying.yang@stats.ox.ac.uk}}
\author[2]{Xing Liu}
\author[1,3]{Robin J.~Evans}
\affil[1]{Department of Statistics, University of Oxford}
\affil[2]{QuantCo}
\affil[3]{Pioneer Centre for SMARTbiomed, University of Oxford}
\date{\today}

\usepackage{amsmath, amssymb, amsthm, mathrsfs, bm, bbm, mathtools}
\usepackage{geometry, extarrows, graphicx, enumitem, booktabs, subfigure, xcolor}
\usepackage{algorithm, algpseudocode}
\algrenewcommand\alglinenumber[1]{\footnotesize #1:}


\usepackage{natbib}
\usepackage{titlesec}
\usepackage{hyperref}
\usepackage{cleveref}
\usepackage{array}
\usepackage{booktabs,makecell}
\usepackage{subcaption}

\usepackage{fullpage}

\usepackage{soul}

\def\ind{\mathbbm{1}}
\def\E{\mathbb{E}}
\def\R{\mathbb{R}}

\def\cN{\mathcal{N}}
\def\cX{\mathcal{X}}

\DeclareMathOperator{\Var}{Var}

\def\cH{\mathcal{H}}

\def\cA{\mathcal{A}}

\newcommand{\iid}{\overset{\mathrm{i.i.d.}}{\sim}}
\newcommand{\Dnml}{\mathcal{D}_{n,m,\ell}}
\newcommand{\tildet}{\widetilde{T}_{n,m,\ell}}
\newcommand{\tildetstar}{\widetilde{T}_{n,m,\ell}^\ast}
\newcommand{\tnml}{T_{n,m,\ell}}
\newcommand{\tnmlstar}{T_{n,m,\ell}^\ast}

\newcommand{\Qcm}{Q_c^m}
\newcommand{\Qtn}{Q_t^n}
\newcommand{\Qhl}{Q_h^\ell}
\newcommand{\Qfml}{Q_f^{m+\ell}}
\newcommand{\opr}{\operatorname{o}_{\Pr}}
\newcommand{\Opr}{\operatorname{O}_{\Pr}}

\newcommand{\tagaligneq}{\refstepcounter{equation}\tag{\theequation}}

\theoremstyle{plain}
\newtheorem{theorem}{Theorem}[section]
\newtheorem{proposition}[theorem]{Proposition}
\newtheorem{lemma}[theorem]{Lemma}
\newtheorem{corollary}[theorem]{Corollary}
\newtheorem{assumption}[theorem]{Assumption}

\newtheorem{remark}[theorem]{Remark}

\begin{document}

\maketitle
\begin{abstract}
Randomized controlled trials (RCTs) are the gold standard for causal inference,  yet practical constraints often limit the size of the concurrent control arm. Borrowing control data from previous trials offers a potential efficiency gain, but naive borrowing can induce bias when historical and current populations differ. Existing test-then-pool (TTP) procedures address this concern by testing for equality of control outcomes between historical and concurrent trials before borrowing; however, standard implementations may suffer from reduced power or inadequate control of the Type-I error rate.

We develop a new TTP framework that fuses control arms while rigorously controlling the Type-I error rate of the final treatment effect test. Our method employs kernel two-sample testing via maximum mean discrepancy (MMD) to capture distributional differences, and equivalence testing to avoid introducing uncontrolled bias, providing a more flexible and informative criterion for pooling. To ensure valid inference, we introduce partial bootstrap and partial permutation procedures for approximating null distributions in the presence of heterogeneous controls. We further establish the overall validity and consistency.  We provide empirical studies demonstrating that the proposed approach achieves higher power than standard TTP methods while maintaining nominal error control, highlighting its value as a principled tool for leveraging historical controls in modern clinical trials.
\end{abstract}
\section{Introduction}
\label{sec:introduction}
A central objective in causal inference is to determine whether an intervention produces a measurable effect in a target population. Although randomized controlled trials (RCTs) are the gold standard, they can be time-consuming and sometimes face ethical or feasibility constraints. Such challenges can usually affect the recruitment of patients into the placebo arm, inflating the variances of estimation and reducing the test power. At the same time, substantial amounts of clinical data are often available prior to the start of a study, particularly for control arms of similar trials \citep{viele2014use}. Thus, a common remedy is to supplement current data with individual patient data from earlier control arms, which would reduce costs, shorten trial duration, and ease recruitment \citep{schmidli2014robust, galwey2017supplementation, schmidli2020beyond}. This follows the valid incentive that, despite each trial’s unique characteristics, multiple previous trials often share similar control treatments \citep{galwey2017supplementation}.

Nevertheless, borrowing historical controls without accounting for potential discrepancies between studies can introduce bias. Differences such as regional bias or assessment bias can be problematic \citep{burger2021use}. To address this, various methods have been proposed to examine the homogeneity before combining historical or external control data with current control arms in settings where fully powered randomized controlled trials are not feasible, including Bayesian approaches such as meta-analytic-predictive priors \citep{schmidli2014robust}, frequentist techniques such as matching \citep{busgang2022selecting}  and hybrid strategies like Bayesian dynamic borrowing \citep{viele2014use}. We refer to \citet{viele2014use}, \citet{galwey2017supplementation} and \citet{kojima2023dynamic} for the review of these methods.

Another mainstream approach is the test-then-pool (TTP) approach \citep{viele2014use, meng2020survey, gao2023pretest}. This approach proceeds in two stages: a fusion test and a causality test.
The fusion test evaluates whether historical and current data are compatible. If there is insufficient evidence to reject compatibility, the historical data are pooled with the current data, and a causal analysis is performed on the combined sample. 

While relatively straightforward, a major criticism of the classic TTP is its limited power in identifying heterogeneity between historical and current controls. As a result, heterogeneous data may be inappropriately pooled, introducing additional bias, and the probability of such happening is \emph{not} controlled by a test targeting the standard point hypotheses.

%that failing to reject null in the fusion test, - the historical and current control are homogeneous - does not actually imply no difference between the means of the two distributions, but showing there is not enough evidence to reject the null  especially if the current placebo data is small. 

To mitigate this problem, \citet{li2020revisit} replaces the fusion test with an equivalence test. This modification provides the Type-I error control when concluding that the two controls are similar. However, since ``similar'' does not imply ``identical'',  the subsequent causality test requires an appropriate adjustment to account for the differences between the control arms to guarantee the Type-I error control for the overall procedure. Such an adjustment was not explicitly derived or formally specified in their framework. 
Moreover, existing TTP methods remain limited to testing the similarity via mean, thereby ignoring the richer information contained in the full distributions. A detailed review of these approaches is provided in \Cref{sec:ttp}.

In this paper, we propose a new TTP framework for fusing control arms, aimed at increasing the power to detect \emph{distributional treatment effects} (DTE) while rigorously controlling the Type-I error rate. Our method addresses a common critique of classic TTP procedures that they can inflate Type-I errors by providing a formal guarantee of validity for the final treatment effect test. Unlike traditional mean-based equivalence tests, our approach captures the full distributional differences between arms via the \textit{Maximum Mean Discrepancy} (MMD), offering a more comprehensive evaluation of treatment effect heterogeneity. We further provide practical guidance for choosing the equivalence margin to achieve desired power levels.

Our key contributions include:
\begin{itemize} 
    \item Extension of TTP to distributional testing, allowing detection of treatment effects beyond mean differences and enabling broader applications.
    \item Partial bootstrap and partial permutation procedures for valid approximation of the null distribution when fused control arms contain differences.
    \item Equivalence-based fusion test using MMD, enabling principled pooling decisions.
    \item Overall TTP validity proof, establishing that the treatment effect test after fusion retains the nominal Type-I error.
\end{itemize}

In \Cref{sec:background}, we provide background review on MMD and two-sample testing based on it, as well as existing TTP procedures. The proposed TTP method is detailed in \Cref{sec:proposed_method}. In \Cref{sec:synthetic_experiment} we present synthetic experiment results, evaluating the performance of the proposed TTP on Type-I error control and demonstrating the affect of different hyperparameter choices. We present the application on a program, \textit{Prospera}, in \Cref{sec:prospera}. We conclude the paper with potential directions in \Cref{sec:discussion}.

\section{Background}
\label{sec:background}

% We use $Q_c$, $Q_h$, $Q_t$ to denote the outcome distributions of the current control, historical control  and treatment groups, respectively. For any distribution $Q$, the superscript $Q^n$ denotes the empirical distributions of $\{x_i\}_{i=1}^n\iid Q$, i.e.~$Q^n = \frac{1}{n}\sum_{i=1}^n \delta_{x_i}$. We use $\xi$ to denote the kernel mean embedding. Our ultimate goal is to test the presence of DTE, namely as a testing problem where the null hypothesis is $\cH_0:Q_c=Q_t$. We refer to it as the \emph{causality} test. For two sets of observations, $\{x_i^{c}\}_{i=1}^m\sim Q_c$ and $\{x_i^{h}\}_{i=1}^\ell\sim Q_h$, the empirical distribution of their pooled sample, $\{x_i^{c}\}_{i=1}^m \bigcup \{x_i^{h}\}_{i=1}^\ell$ is denoted  $\Qfml= \frac{m}{m+\ell}\Qcm+\frac{\ell}{m+\ell}\Qhl$. We also have observations from treatment group, $\{x_i^{t}\}_{i=1}^n$.

Let $Q_c$, $Q_h$, $Q_t$ be probability measures on $\R$, which denote the outcome distributions of the current control, historical control and treatment groups, respectively. We assume we observe independent random variables $\{x_i^{c}\}_{i=1}^m\sim Q_c$, $\{x_i^{h}\}_{i=1}^\ell\sim Q_h$ and $\{x_i^{t}\}_{i=1}^n \sim Q_t$. For any distribution $Q$, its empirical distribution based on independent samples $\{x_i\}_{i=1}^n \sim Q$ is denoted by $Q^n = \frac{1}{n} \sum_{i=1}^n \delta_{x_i}$, where $\delta_x$ denotes the Dirac measure at $x$. The empirical distribution based on $\{x_i^{c}\}_{i=1}^m$ is denoted by $Q_c^m = \frac{1}{m} \sum_{i=1}^m \delta_{x_i^c}$, and similarly for $Q_t^n$ and $Q_h^\ell$. Our main goal is to test for statistically significant difference between the treatment group $Q_t$ and the control group $Q_c$, possibly after augmenting by the historical control $Q_h$. This is phrased as testing the DTE by setting $\cH_0:Q_c=Q_t$. We refer to it as the \emph{causality} test. For two sets of observations, $\{x_i^{c}\}_{i=1}^m\sim Q_c$ and $\{x_i^{h}\}_{i=1}^\ell\sim Q_h$, the empirical distribution of their pooled sample, $\{x_i^{c}\}_{i=1}^m \bigcup \{x_i^{h}\}_{i=1}^\ell$, is denoted  $\Qfml= \frac{m}{m+\ell}\Qcm+\frac{\ell}{m+\ell}\Qhl$.

\subsection{Distributional treatment effect}
% \xl{
% 1. Our main goal is to test for statistically significant difference between the treatment group $Q_t$ and the control group $Q_c$, possibly after incorporating the historical control $Q_h$. In the TTP literature, the causal effect is typically defined to be the average treatment effect (ATE) by targeting the null hypothesis $\cH_0: \E_{X\sim Q_c}[X] = \E_{Y \sim Q_t}[Y]$. While the ATE is the most widely studied causal estimand in the literature, it is limited to only the first moments and ignore higher moments. 
% 2. For this reason, we focus on DTE in this work.
% 3. List advantages of DTE compared to ATE
% 4. Ours is the first TTP that concerns DTE. Moreover, by choosing suitable kernels, our test can be specialized into detecting ATE too, thus unifying ATE and DTE inference.
% } 

While the average treatment effect (ATE) is the most widely studied causal estimand, it summarizes impact only through differences in means. This focus can obscure important heterogeneity: two treatments may have the same mean effect but very different impacts on the distribution of outcomes. Thus, ATE provides only limited insights. In contrast, DTEs address this limitation by characterizing how an intervention shifts the entire outcome distribution, rather than just its first moment. Increasing attention is being paid to distributional treatment effects, which offer more comprehensive information beyond the mean \citep{abadie2002bootstrap,kennedy2023semiparametric}.

A growing body of work develops nonparametric methods for estimating DTEs. One strand leverages kernel embeddings of probability distributions, which represent distributions in a reproducing kernel Hilbert space (RKHS). This approach allows treatment effects to be evaluated through differences between embedded distributions, making it possible to capture complex shifts without imposing strong parametric assumptions. Kernel-based methods also provide a natural way to compare entire distributions, and can be integrated with flexible machine learning models \citep{muandet2017kernel, muandet2021counterfactual, park2021conditional}. 

By moving beyond comparing only the first moments, estimation of this kind offers a more flexible and informative framework, particularly in settings where heterogeneity or tail behavior is crucial, or when the outcomes are of high-dimensional.

\begin{remark}[Connection with ATE]
An appealing feature of the proposed kernel two-sample testing framework is its flexibility. By choosing a characteristic kernel, the method captures full distributional differences and thus enables inference on distributional treatment effects. At the same time, if we choose the linear kernel $k(x,y) = x^\top y$, the framework reduces to a test for mean differences, recovering the classical average treatment effect. This dual capability highlights that our approach unifies ATE and DTE inference within a single methodological framework. 
\end{remark}

\subsection{MMD and two-sample testing}
\label{sec:MMD_background}
Our TTP framework requires a notion of dissimilarity that quantifies distributional differences. We use the Maximum Mean Discrepancy \citep[MMD;][]{gretton2006kernel,gretton2012kernel}, which we review next. 

Let $k:\cX \times \cX \rightarrow\R$ be a bounded, positive definite kernels. Denote by $\cH_k$ be its associated Reproducing Kernel Hilbert Space \citep[RKHS;][]{berlinet2011reproducing}, equipped with the RKHS inner product $\langle \cdot, \cdot \rangle_{\cH_k}$. The MMD between two distributions $Q$ and $P$ is the maximal difference between their expectations of all functions in the unit ball of $\cH_k$, defined as
\begin{align*}
    D(Q, P) \coloneqq \sup_{f \in \cH_k: \| f \|_{\cH_k} \leq 1} \big| \E_{X \sim Q}f(X) - \E_{Y \sim P}f(Y) \big|
    ,
\end{align*}
where $\| f\|_{\cH_k} = \langle f, f \rangle_{\cH_k}^{1/2}$ is the induced RKHS norm. Under regularity conditions, the MMD can be rewritten as the RKHS-norm of the difference between the so-called \emph{kernel mean embeddings} of the two distributions $Q$ and $P$. For a probability measure $R$, we define its kernel mean embedding as $\xi_R \coloneqq \E_{X \sim R}\,k(\cdot, X)$. Throughout, we will assume the kernel $k$ satisfies the following condition, which guarantees the existence of $\xi_Q, \xi_P$ \citep[Lemma 3]{gretton2012kernel}.
\begin{assumption}[Finite first moment]
\label{ass:kernel_moment}
    The kernel $k$ satisfies $\E_{X \sim R}\sqrt{k(X, X)} < \infty$, for both $R = Q$ and $R=P$.
\end{assumption} 
Under \Cref{ass:kernel_moment}, the \emph{squared} MMD admits the following equivalent forms \citep[Lemma 4, 6]{gretton2012kernel}
\begin{align*}
    D^2(Q, P) 
    &= \big\| \xi_Q - \xi_P \big\|_{\cH_k}^2
    = \E_{X, X' \sim Q}k(X, X') + \E_{Y, Y' \sim P}k(Y, Y') - 2\E_{X \sim Q, Y \sim P}k(X, Y) .
\end{align*}
We further assume that $k$ is \emph{characteristic} \citep{sriperumbudur2011universality} throughout the paper, meaning that the mapping $Q \mapsto \xi_Q$ is injective. In this case, the associated MMD is a statistical metric. In particular, MMD separates distributions, meaning that $D(Q, P) = 0$ if and only if $Q = P$ \citep[Theorem 5]{gretton2012kernel}. Most common kernels, such as Gaussian kernels, Inverse Multi-Quadric (IMQ) kernels, and Mat\'{e}rn kernels, are characteristic. The former two kernels also satisfy \Cref{ass:kernel_moment} for all $Q, P$, since they are bounded.

The above identities also suggest a natural finite-sample estimator for MMD using Monte-Carlo estimation. Given random samples $\{X_i\}_{i=1}^n$ and $\{Y_j\}_{j=1}^m$ drawn independently from $Q$ and $P$, respectively, the MMD can be estimated by the following biased but consistent estimator, formed by replacing the expectations with empirical averages
\begin{align}
    D^2(Q^n, P^m) =
    \frac{1}{n^2}\sum_{i=1}^n\sum_{j=1}^n k(X_i, X_j) + \frac{1}{m^2}\sum_{i=1}^m \sum_{j=1}^m k(Y_i, Y_j) - \frac{2}{nm}\sum_{i=1}^n \sum_{j=1}^m k(Y_i, Y_j).
\label{eq:mmd_v_stat}
\end{align}
It can be shown that \eqref{eq:mmd_v_stat} can be rewritten as a \emph{two-sample V-statistic of degree $(2,2)$} \citep{gretton2012kernel,shekhar2022permutation}, whose asymptotic properties are well-known \citep[Chapter 5]{serfling2009approximation}.

\begin{remark}[U-statistics]
An alternative estimator in the literature is the following estimator
\begin{align}
    U_{nm}=
    \frac{1}{n(n-1)}\sum_{i=1}^n\sum_{\substack{j=1 \\ j\neq i}}^n k(X_i, X_j) + \frac{1}{m(m-1)}\sum_{i=1}^m \sum_{\substack{j=1 \\ j\neq i}}^m k(Y_i, Y_j) - \frac{2}{nm}\sum_{i=1}^n \sum_{j=1}^m k(X_i, Y_j).
    \label{eq:mmd_u_stat}
\end{align}
Similarly to \eqref{eq:mmd_v_stat}, the estimator above can be rewritten as a two-sample U-statistic \citep{hoeffding1948class}. U-statistics share similar statistical properties as V-statistics and have the advantage of being unbiased \citep[Chapter 5]{serfling2009approximation}. For these reasons, \eqref{eq:mmd_u_stat} is the more commonly used estimator in the MMD literature. However, U-statistics can take negative values, whereas V-statistics are always non-negative. As will become clear in \Cref{sec:equivalence_test}, non-negativity is essential for the fusion test used in our proposed method. Therefore, we will focus on V-statistics in this work, and defer a discussion on extensions to U-statistics to \Cref{app:u_stats}.
\end{remark}

MMD has been widely applied in two-sample testing for the homogeneity of two distributions \citep{gretton2006kernel,gretton2012kernel, schrab2023mmd,shekhar2022permutation,chau2024credal}. To test the null hypothesis $\cH_0: Q = P$, MMD checks first the separation property to reformulate the null hypothesis to the equivalent form $\cH_0: D(Q, P) = 0$. It then rejects the null hypothesis if the estimate $D(Q^n, P^m)$ exceeds a critical value, which is the $(1-\alpha)$-quantile of the distribution of $D(Q^n, P^m)$ under $\cH_0$, for a prescribed significance level $\alpha$.

One major challenge in MMD tests is that the critical value is typically intractable and thus needs to be estimated. Common approaches include permutation and bootstrapping, of which both aim to construct empirical samples from the null distribution. The permutation method proceeds by randomly permuting the pooled sample $\big\{Z_\ell\big\}_{\ell=1}^{n+m}=\big\{X_i\big\} \cup\big\{Y_j\big\}$ to form the permuted sample $\big\{Z_\ell^\pi\big\}_{\ell=1}^{n+m}$, and computing the permuted MMD statistic 
\begin{align}
    T_\pi\coloneqq D^2\bigg(\frac{1}{n}\sum_{i=1}^n \delta_{Z_i^\pi}, \frac{1}{m}\sum_{i=n+1}^{n+m} \delta_{Z_i^\pi} \bigg).
\end{align}
Repeating this over $B$ random permuted samples yields the empirical quantity $\big\{T_{\pi^{(b)}}\big\}_{b=1}^B$. The critical value  $\hat{q}^{m,n}_{1-\alpha}$ is then given by the Monte Carlo estimate based on this sample. 
% The test $\phi: \R^m \times \R^n \rightarrow\{0,1\}$ is defined such that $\phi(Q^m,P^n)=1$ if $\cH_0$ is rejected and $\phi(Q^m,P^n)=0$ otherwise, namely
% \begin{align*}
%     \phi(Q^m,P^n)=\mathbbm{1}\Big\{D^2(Q^m,P^n)>\hat{q}_{1-\alpha}^{m,n}\Big\}.
% \end{align*}

Alternatively, the critical value can be estimated using bootstrapping approaches, with the most common method being wild bootstrapping \citep{shao2010dependent,schrab2023mmd}, which constructs the bootstrap samples using Rademacher weights. In our work, we instead follow \citet{liu2026kerneltestsequivalence} to use Efron's bootstrapping, where the bootstrapped statistics are formed by drawing samples from $\{Z_l\}_{l=1}^{n+m}$ with replacement. The critical value is then the $(1-\alpha)$-quantile based on the bootstrapped statistics. As argued in \citet[Section 2.2]{liu2026kerneltestsequivalence}, Efron's bootstrapping gives similar empirical performance while allowing existing theoretical tools to be used to prove statistical properties for their equivalence MMD test.

\subsection{Existing mean-based TTP }
\label{sec:ttp}
We provide a more detailed review of existing TTP methods. As introduced in \Cref{sec:introduction}, TTP is a general two-stage framework consisting of a fusion test and a causality test. The fusion test determines whether a historical control group with outcome distribution $Q_h$ is sufficiently compatible with the current control group $Q_c$ to justify pooling. If compatibility is not rejected, the external and current controls are merged, and the resulting pooled control is compared with the treatment group  $Q_t$ in the causal stage.

The classic TTP does so by performing a statistical test targeting $\check{\cH}_0^{f}: \mu_c = \mu_h$, where $\mu_c$, $\mu_h$ are the means of $Q_c, Q_h$, respectively. If the fusion test fails to be rejected, then the external control group is merged with the current control group to form a pooled group.  The causality test then compares the resulting control group---either alone or pooled---with the treatment group, $Q_t$, to assess the presence of a treatment effect. This is formalized by performing the causality test targeting $\check{\cH}_0: \mu_c = \mu_t$.

While this approach is simple, it is often insufficient to conclude equivalence between the control groups. Indeed, failure to reject $\check{\cH}_0^{f}$ does not imply that $\mu_c$ and $\mu_h$ are equivalent; instead, it can simply be a result of lack of statistical power. Consequently, the TTP framework can incorrectly merge heterogeneous historical control arms into the causal analysis, especially when the current control size $m$ is small relative to the treatment size $n$. This, in turn, can then lead to inflated Type-I error in the causality test \citep{viele2014use}. 

To address this limitation, \citet{li2020revisit} proposed to replace the fusion test with an equivalence test, targeting instead the null hypothesis $ \check{\cH}_0^{f}: |\mu_c - \mu_h| \geq \delta$ for some pre-specified margin $\delta > 0$. Equivalence testing is more natural for assessing the similarity between the control arms, because it provides probabilistic guarantee on the error rate of falsely merging heterogeneous historical controls. However, this equivalence-based TTP test till relies solely on \emph{mean} comparisons and does not account for differences in the broader outcome distributions, which can be important when the goal is to assess effects beyond the mean.

\subsection{Warm-up: Distributional TTP}
The TTP framework is general in nature: any statistical test can be used in the fusion stage. To align with the classic TTP procedure and facilitate comparison, we extend it here from mean-based to distribution-based testing, allowing the framework to target the DTE instead of the ATE. This formulation is not proposed in prior work; rather, we introduce it to formalize a ``classic TTP'' analogue at the distributional level, which will serve as a conceptual baseline and benchmark for our later developments. Specifically, we define the causal null hypothesis as
\begin{align*}
    \cH_0: Q_c=Q_t \quad \text{against} \quad \cH_1:Q_c\neq Q_t.
\end{align*}
When using an MMD two-sample test with a characteristic kernel, this is equivalent to 
\begin{align*}
   \cH_0: D(Q_c, Q_t)=0 \quad \text{against} \quad  \cH_1: D(Q_c, Q_t) > 0.
\end{align*}

Without merging from external control group, the causality test uses merely $\Qcm$ and $\Qtn$. When one sample size is small---here we assume $m\ll n$---the test may have insufficient power to detect differences in treatment effects. This motivates fusing the current control $\Qcm$ with an external or historical control, whose response distribution $Q_h$ is sufficiently similar to $Q_c$. 

The fusion test assesses
\begin{align*}
    \tilde{\cH}_0^{f}: Q_c=Q_h \quad \text{against} \quad  \tilde{\cH}_1^{f}: Q_c\neq Q_h.
\end{align*}
or equivalently,
\begin{align}
    \tilde{{\cH}}_0^{f}: D\big(Q_h,Q_c\big)=0 \quad \text{against} \quad  \tilde{\cH}_1^{f}: D\big(Q_h,Q_c\big) > 0.
     \label{eq:classic_fusion_hypothesis}
\end{align}
The corresponding test function for the fusion stage is
\begin{align}
    \phi_f^0\big(\Qcm,\Qhl\big) = \mathbbm{1}\Big\{D^2\big(\Qcm,\Qhl\big)> \hat{q}_{1-\alpha_f}^{m,\ell}\Big\}
    ,
    \label{eq:classic_fusion_test}
\end{align}
where $\alpha_f$ is the prescribed significance level. Here, $\hat{q}_{1-\alpha_f}^{m,\ell}$ is the permutation-based $(1-\alpha_f)$-quantile of the test statistic.

The intuition is straightforward: if  $\tilde{\cH}_0^f$ is rejected, $\Qhl$ is excluded from pooling, and the causality test remains to be $\phi^0\big(\Qcm,\Qtn\big)$. Otherwise, the fused control-arm response distribution is $\Qfml = \frac{m}{m+\ell}\Qcm+\frac{\ell}{m+\ell}\Qhl$, which is then compared to $\Qtn$ in the causality testing stage using an MMD test based on $D^2\big(\Qfml, \Qtn\big)$.

We denote the overall TTP decision rule by $\phi^0\big(\Qcm,\Qhl,\Qtn\big): \R^m \times \R^l \times \R^n \to \{0,1\}$, where $\phi^0=1$ indicates rejection of $\cH_0$. The function $\phi^0$ is defined conditional on the fusion stage test function $\phi_f^0$, written as
\begin{align*}
    \phi^0\big(\Qcm,\Qhl,\Qtn\big)= \begin{cases}
   \mathbbm{1}\Big\{D^2\big(\Qcm,\Qtn\big)>\hat{q}_{1-\alpha}^{m,n}\Big\}\,, & \phi_f^0\big(\Qcm,\Qhl\big)=1,\\
    \\[-6pt]
    \mathbbm{1}\Big\{D^2\big(\Qfml, \Qtn\big)>\hat{q}_{1-\alpha}^{m+\ell,n}\Big\}\,, & \phi_f^0\big(\Qcm,\Qhl\big)=0.
\end{cases}
\end{align*} 
In both cases, the causality test decision at significance level $\alpha$ is made via the permutation-based test.

This formulation generalizes the classical TTP to a distributional setting. However, the fundamental issue remains: the fusion test $\tilde{\cH}_f$ may be underpowered when the current control sample size $m$ is small, leading to an incorrect merge and inflated Type-I error in the causal stage. This limitation motivates the development of our proposed method.

\section{Method}
\label{sec:proposed_method}
To address the aforementioned problems, we propose a new TTP procedure in this paper. For the fusion stage, we replace the classic two-sample test with an equivalence test. Conditional on the fusion test outcome, the causality test is conducted using either a traditional permutation test if the historical sample is not fused, or a partial bootstrap/permutation test otherwise. The complete procedure is given in \Cref{alg:equiv_ttp}. 

\begin{algorithm}[t]
    \caption{Equivalence TTP}
    \label{alg:equiv_ttp}
    \begin{algorithmic}[1]
        % \State {\bfseries Inputs:} 
        \Require Data $\Qcm$, $\Qhl$ and $\Qtn$;
        equivalence radius $\theta > 0$; 
        test levels $\alpha$ for the causality test and $\alpha_f$ for the fusion test.
        \State  Conduct equivalence test  $\cH_0^f: D\big(Q_h,Q_c\big) \geq \theta$ at level $\alpha_f$.
        \If{$\cH_0^f$ is \textbf{not} rejected}
            \State Conduct a  permutation test between $\Qcm$ and $\Qtn$ for $\cH_0$ at level $\alpha$.
        \Else
            \State Form the pooled control sample 
            \[
                \Qfml = \frac{m}{m+\ell} \Qcm + \frac{\ell}{m+\ell} \Qhl.
            \]
            \State Conduct a partial bootstrap or partial permutation causality test between $\Qfml$ and $\Qtn$ for $\cH_0$ at level $\alpha$.
        \EndIf
    \end{algorithmic}
\end{algorithm}

Unlike the equivalence-based TTP proposed by \citet{li2020revisit}, our method flexibly detects treatment effects at the distributional level, is not restricted to normal distribution, and does not require the variance to be known. In addition, the proposed partial permutation and partial bootstrap procedures ensure the asymptotic validity and consistency of the entire framework, even when the fused control arms are not identical. These properties provide formal theoretical guarantees for Type-I error control of the overall framework, which were not established in \citet{li2020revisit}.

\subsection{Fusion test: MMD equivalence testing}
\label{sec:equivalence_test}
We introduce the fusion test in this section. We propose to test
\begin{align}
    \cH_0^f: D(Q_c,Q_h)\geq \theta \quad\text{against}\quad \cH_1^f:D(Q_c,Q_h)<\theta,
    \label{eq:equiv_fusion_hypothesis}
\end{align}
where $\theta>0$ is the \emph{equivalence radius}.

The benefits of testing the above hypotheses are two-fold. First, unlike \eqref{eq:classic_fusion_hypothesis} where the null hypothesis is that $Q_c$ and $Q_h$ are identical, we employ an equivalence testing setting, where the null hypothesis is that they differ by at least $\theta$ in MMD. With this formulation, we reject only when there is statistical evidence for $D(Q_c,Q_h)<\theta$. This ensures control over the probability of an ``incorrect merge'', where an external control $Q_h$ is distinct from $Q_c$ but the fusion test still suggests merging (namely, not rejecting the $\cH_0^f$ in \eqref{eq:classic_fusion_hypothesis}) due to insufficient power. Meanwhile, we use MMD as a notion of dissimilarity
between $Q_c$ and $Q_h$, which captures not only differences in their means, but in their entire distributions. 

To test \eqref{eq:equiv_fusion_hypothesis}, we employ the MMD equivalence test proposed in \citet[Section 4.2]{liu2026kerneltestsequivalence}. This test uses the test statistic
\begin{align*}
    \Delta_{m,\ell}^f \,\coloneqq\, \theta - D(Q_c^m, Q_h^\ell) ,
\end{align*}
and rejects $\cH_0^f$ in \eqref{eq:equiv_fusion_hypothesis} if $\Delta_{m,\ell}^f > q^{m,\ell}_{1-\alpha_f}$, where $q^{m,\ell}_{1-\alpha_f}$ is the $(1-\alpha_f)$-th quantile of the distribution of $S_{m,\ell} \coloneqq D(Q_c^m, Q_c) + D(Q_h^\ell, Q_h)$.

The intuition behind $S_{m,\ell}$ is that, under $\cH_0^f: D\big(Q_h,Q_c\big) \geq \theta$, it upper bounds $\Delta_{m,\ell}^f$ through a triangle inequality. It can also be shown that $S_{m,\ell}$ is close to 0 with high probability under $\cH_0^f$. Hence, the MMD equivalence test rejects the equivalence null hypothesis if $\Delta_{m,\ell}^f$ is sufficiently large, or, equivalently, if $D(Q_c^m, Q_h^\ell)$ is sufficiently smaller than $\theta$. Using this intuition, \citet[Theorem 8]{liu2026kerneltestsequivalence} establishes the validity and consistency of this equivalence MMD test. We restate their result below with slight adaptations to suit our setting. Its proof is in \Cref{sec:proof_mmd_equiv_test}.

\begin{theorem}[Calibration and consistency of MMD equivalence test]
    \label{thm:mmd_equiv_test}
    Let $\alpha_f \in (0, 1)$. Denote by $q_{1-\alpha_f}^{m,\ell}$ the $(1-\alpha_f)$-th quantile of the distribution of $S_{m,\ell} = D(Q_c^m, Q_c) + D(Q_h^\ell, Q_h)$. Assume $\E_{X, X' \sim Q_c}\big| k(X, X') \big|^2 < \infty$ and $\E_{Y, Y' \sim Q_h}\big| k(Y, Y') \big|^2 < \infty$. Then there exists $a \in (0, \alpha_f]$ such that
    \begin{align*}
        \lim_{n \to \infty} \Pr\big(\Delta_{m,\ell}^f > q_{1-\alpha_f}^{m,\ell}\big)
        \;=\;
        \begin{cases}
            0 , &D\big(Q_h,Q_c\big) > \theta , \\
            a , &D\big(Q_h,Q_c\big) = \theta , \\
            1 , &D\big(Q_h,Q_c\big) < \theta .
        \end{cases}
    \end{align*}
\end{theorem}

However, a challenge is that $q_{1-\alpha_f}^{m,\ell}$ is often intractable, because the distribution of $S_{m,\ell}$ is typically unknown. \citet{liu2026kerneltestsequivalence} proposed to use a bootstrap procedure to approximate this unknown quantile. Define 
\begin{align*}
    D^2_W(Q_c^m)
    \coloneqq
    \frac{1}{m^2} \sum_{i=1}^m \sum_{j=1}^m (W_{i} - 1) (W_{j} - 1) k(x_i, x_j)
    ,
\end{align*}
where $W = (W_{1}, \ldots, W_{n}) \sim \textrm{Multinomial}(n; 1/n, \ldots, 1/n)$, and similarly define $D^2_W(Q_h^\ell)$. We then define the bootstrap statistic
\begin{align}
    S^b_{m,\ell}
    \,\coloneqq\, D_{W^b}(Q_c^m) + D_{\widetilde{W}^b}(Q_h^\ell)
    ,
    \label{eq:mmd_equiv_boot_sample}
\end{align}
where $\{ W^b \}_{b=1}^B$ and $\{ \widetilde{W}^b \}_{b=1}^B$ are drawn independently from $\textrm{Multinomial}(n; 1/n, \ldots, 1/n)$. \citet[Lemma 7]{liu2026kerneltestsequivalence} shows that $q^{m,\ell}_{1-\alpha_f}$ can be approximated by the $(1-\alpha_f)$-th quantile of the distribution of the bootstrap sample $\{ S_{m,\ell}^b \}_{b=1}^B$, namely
\begin{align}
    q^{m,\ell}_{1-\alpha_f, B}
    =
    \inf\Big\{u \in \R: \, \frac{1}{B} \sum_{b = 1}^{B} \mathbbm{1}\big\{ S_{m,\ell}^b \leq u \big\} \geq 1 - \alpha \Big\}
    .
    \label{eq: bootstrap quantile}
\end{align}
To summarize, the MMD equivalence test rejects $\cH_0^f$ if $\Delta_{m,\ell}^f > q^{m,\ell}_{1-\alpha_f, B}$. We will denote this test by $\phi_f: \R^m \times \R^\ell \to \{0, 1\}$, where $\phi_f\big(\Qcm, \Qhl\big) = \ind\Big\{ \Delta_{m,\ell}^f > q^{m,\ell}_{1-\alpha_f} \Big\}$.

\subsection{Causality test: partial bootstrap and partial permutation}
Let $\phi_\ast: \R^m \times \R^\ell \times \R^n \to \{0, 1\}$ denote the causality test for $\cH_0: D\big(Q_c, Q_t\big) = 0$, so that $\phi_\ast(\Qcm, \Qhl, \Qtn) = 1$ if $\cH_0$ is rejected, and $0$ otherwise. In particular, conditional on $\phi_f\big(\Qcm, \Qhl\big) = 0$, no merging is performed, and hence the causality test $\phi_\ast\big(\Qcm, \Qhl, \Qtn\big)$ does \emph{not} depend on $\Qhl$.

One of the key step in the causal stage is to approximate the null distribution of the test statistic under $\cH_0$, so that we can estimate the testing threshold. Unlike traditional TTP, the $\Qhl$ we decide to merge can be different from $\Qcm$ (by at most $\theta)$. Consequently, the merged control-arm distribution $\Qfml$ may also differ from $\Qcm$. This creates a problem for directly applying a permutation test between $\Qfml$ and $\Qtn$, as such permutation approximates the distribution of $D^2\big(\Qfml, \Qtn\big)$ when $Q_f=Q_t$, where $Q_f$ is the population-level mixture 
\begin{align*}
    Q_f=\frac{m}{m+\ell}Q_c+ \frac{\ell}{m+\ell}Q_h.
\end{align*}

 The null distribution of $D^2\big(\Qfml, \Qtn\big)$ coincides with that under $Q_f=Q_c$ only when $Q_h=Q_c$. If $Q_h\neq Q_c$, even by a small margin, the direct permutation approach no longer yields the correct null distribution, leading to potential size distortion. We thus propose the following methods to mitigate this problem. For the convenience of the following discussion, we impose the assumption of fixed ratios:

 \begin{assumption}[Fixed sample ratios]
 \label{ass:fixed_ratios}
The samples $\Qcm$, $\Qhl$ and $\Qtn$ are independent with $c_1\coloneqq\frac{m}{n}\in(0,1]$ and $c_2\coloneqq\frac{\ell}{n}\in(0,1]$, where $c_1$ and $c_2$ are constants.
 \end{assumption}
 
\begin{remark}[Generalization of fixed ratios]
\label{remark:generalize_fixed_ratio}
Although we assume fixed sample size ratios $\ell/n$ and $m/n$, the results below extend to cases where these ratios converge to finite limits, as long as the limits are greater than $0$. The proofs remain valid under this relaxation by adding error terms controlling the convergence of ratios. We provide a simple example in \Cref{sec:example_varying_ratio}.
\end{remark}
\subsubsection{Partial bootstrap}
To approximate the distribution of $D^2\big(\Qfml, \Qtn\big)$ under $\cH_0$ in the fused setting, we propose the \emph{partial bootstrap} algorithm. This procedure resamples the current control and treatment groups from the same source distribution $Q_c$ under the null, while independently resamples the historical control group $Q_h$. This preserves the correct dependence structure of the fused control under $\cH_0$ even when $Q_h\neq  Q_c$. The details are provided in \Cref{alg:partial_bootstrap}.

\label{sec:partial_bootstrap}

\begin{algorithm}[!tb]
    \caption{Partial Bootstrap.}
    \label{alg:partial_bootstrap}
    \begin{algorithmic}[1]
        \Require Data $Q_c^m, Q_h^\ell, Q_t^n$; number of bootstrap replicates $B$; test level $\alpha$.
       \For{$b \gets 1$ to $B$}
       \State Bootstrap $m$, $n$ samples from $Q_c^m$ to get $Q_{c,b}^m$, $Q_{t,b}^n$. \label{line:bootstrap}
       \State Bootstrap $\ell$ samples from $Q_h^\ell$ to get $Q_{h,b}^\ell$.
       \State Form $Q_{f,b}^{m+\ell} = \frac{m}{m+\ell}Q_{c,b}^m+ \frac{\ell}{m+\ell}Q_{h,b}^\ell$.
       \State Calculate $\tnmlstar\coloneqq D^2\big(Q_{t,b}^n,Q_{f,b}^{m+\ell}\big)$ and $\tildetstar \coloneqq D^2\big(Q_{c,b}^m,Q_{f,b}^{m+\ell}\big).$
       \EndFor
       
    \State \Return The empirical $(1-\alpha)$-quantile  
    \begin{align}
        \hat{q}_{1-\alpha,B}^{m+\ell,n}=\inf\Bigg\{u\in\mathbb{R}:\frac{1}{B}\sum_{b=1}^B\mathbbm{1}\bigg\{\sqrt{n}\Big[D^2\big(Q_{t,b}^n,Q_{f,b}^{m+\ell}\big)-D^2\big(Q_{c,b}^m,Q_{f,b}^{m+\ell}\big)\Big]\leq u\bigg\}\geq 1-\alpha\Bigg\}.
        \label{eq:partial_boot_quantile}
    \end{align}
    \end{algorithmic}
\end{algorithm}

The test statistic in partial bootstrap is 
\begin{align*}
    \Delta \coloneqq \sqrt{n}\Big(D^2\big(\Qfml,Q_t^n\big) - D^2\big(\Qfml,Q_c^m\big)\Big), 
\end{align*}
and we reject $\cH_0$ at test level $\alpha$ when $\Delta > \hat{q}_{1-\alpha,B}^{m+\ell,n}$, with $\hat{q}_{1-\alpha,B}^{m+\ell,n}$ defined in \eqref{eq:partial_boot_quantile}.

For notational brevity, we write $\tnml = D^2\big(\Qfml,\Qtn\big)$ and $\tildet  = D^2\big(\Qfml,\Qcm\big)$, so that $\Delta =\sqrt{n}\big(\tnml-\tildet\big)$. Furthermore, let $\tnmlstar\coloneqq D^2\big(Q_{f,b}^{m+\ell},Q_{t,b}^{n}\big)$ be the partial bootstrap statistic constructed as in \Cref{alg:partial_bootstrap}, and let $\tildetstar\coloneqq D^2\big(Q_{f,b}^{m+\ell},Q_{c,b}^{m}\big)$ be its centering.  In particular, $\tildet$ is ``partially'' bootstrapped because, although $Q_{f,b}^{m+\ell}$ is sampled with replacement from $Q_f^{m+\ell}$, the other sample $Q_{t,b}^n$ is bootstrapped from the current control arm $Q_c^m$, instead of the standard choice $Q_t^n$. The latter construction is crucial to ensure that the proposed bootstrapping procedure can approximate the null distribution of $\Delta$, even if $Q_c \neq Q_h$. Let ${\Pr}^\ast$ be the bootstrap probability conditional on the observed data $\Dnml\coloneqq\{x_i^{t}\}_{i=1}^n\,\cup\, \{x_i^{c}\}_{i=1}^m\,\cup\, \{x_i^{h}\}_{i=1}^\ell$ that governs the resampling steps. Define
\begin{align*}
    \Delta^\ast = \sqrt{n}\big(\tnmlstar - \tildetstar\big).
\end{align*}

\begin{proposition}[Weak convergence of the partial bootstrap under non-identical control groups]
\label{prop:bootstrap_convergence} 
Suppose \Cref{ass:kernel_moment} holds under $R = Q_c, Q_h, Q_t$, and \Cref{ass:fixed_ratios} is satisfied.  Under the causal null $\cH_0:Q_c=Q_t$ and assuming $Q_c\neq Q_h$, we have
\begin{align*}
   \Delta \rightsquigarrow_{\Pr} \cN\big(0,4(1+1/c_1)\sigma_c^2\big)\quad \text{and} \quad \Delta^\ast \rightsquigarrow_{\Pr} \cN\big(0,4(1+1/c_1)\sigma_c^2\big) 
   ,
\end{align*}
where \begin{align*}
    \sigma_c^2 &\coloneqq\Var_{X\sim Q_c}\Big[\E_{Y\sim Q_f}k(X, Y) - \E_{Z\sim Q_c}k(X, Z)\Big].
\end{align*}
We thus obtain
\begin{align*}
    \sup_{z\in\R} \Big|{\Pr}^\ast\big\{\Delta^\ast\leq z\big\} - \Pr\big\{\Delta\leq z\big\} \Big|\overset{a.s.}{\longrightarrow}0.
\end{align*}
\end{proposition}

We provide the proof in \Cref{sec:proof_bootstrap_convergence}. Note that the ratio $c_1 = m/n$ governs the asymptotic variance of $\Delta$: larger values of $c_1$ yield a less variable distribution. This matches the intuition that a larger proportion of current control observations improves stability. Meanwhile, $Q_f$ and $\Qfml$ affect the limiting distributions of $\Delta$ and $\Delta^\ast$ via $\sigma_c^2$  as shown in its definition.

With \Cref{prop:bootstrap_convergence}, we obtain the validity and consistency of partial bootstrap, presented in \Cref{thm:partial_boots_validity,thm:partial_boots_consistency}. Their proofs are deferred to \Cref{sec:proof_partial_boots_validity,sec:proof_partial_boots_consistency}, respectively.

\begin{theorem}[Partial bootstrap validity under non-identical control groups]
\label{thm:partial_boots_validity}
    Suppose \Cref{ass:kernel_moment} holds under $R = Q_c, Q_h, Q_t$, and \Cref{ass:fixed_ratios} is satisfied. Let $\hat{q}_{1-\alpha,B}^{m+\ell,n}$ be the partial bootstrap critical value defined in \eqref{eq:partial_boot_quantile}. Under the causal null hypothesis $\cH_0:Q_c=Q_t$ and assuming $Q_c\neq Q_h$, the test has asymptotic level $\alpha$:
\begin{align*}
    \limsup_{n\to\infty}\lim_{B\to\infty}\Pr \Big\{\Delta>\hat{q}_{1-\alpha,B}^{m+\ell,n}\,\Big|\,\cH_0\Big\}\leq\alpha.
\end{align*}
\end{theorem}

\Cref{thm:partial_boots_validity} assumes $Q_c \neq Q_h$. When $Q_c=Q_h=Q_t$, $\sigma_c^2$ in \Cref{prop:bootstrap_convergence} becomes $0$, resulting in the limiting distribution  being trivial. Our empirical results in \Cref{fig:test_statistic_approximation} suggest that the resulting critical values are greater than or equal to the true ones, implying that the test may be conservative but remains valid. In practice, the case $Q_h=Q_c$ is less common than $Q_t=Q_c$, since temporal shifts or changes in the environment or population typically separate $Q_h$ from $Q_c$. A theoretical analysis of this special case remains an interesting direction for further investigation.

We next provide the consistency of the partial bootstrap test, irrespective of whether $Q_c$ and $Q_h$ are identical or not. Let $\gamma \coloneqq c_1/(c_1+c_2) $, which equals to $m/(m+\ell)$ under \Cref{ass:fixed_ratios}. We denote $\cos \beta \coloneqq  \langle \xi_h-\xi_c\,,\, \xi_t-\xi_c\rangle/\big(\|\xi_h-\xi_c\|\|\xi_t-\xi_c\|\big)$ when $Q_c\neq Q_h$. 

\begin{theorem}[Partial bootstrap test consistency]
\label{thm:partial_boots_consistency} Suppose \Cref{ass:kernel_moment} holds under $R = Q_c, Q_h, Q_t$, and \Cref{ass:fixed_ratios} is satisfied. Suppose $Q_t\neq Q_c$. Regardless of whether $Q_c=Q_h$, we have
\begin{align*}
    \lim_{n\to\infty}\lim_{B\to\infty} {\Pr}\Big\{\Delta>\hat{q}_{1-\alpha,B}^{m+\ell,n}\,\big|\,\cH_1\Big\} = 1,
\end{align*}
if and only if $Q_c=Q_h$, or $Q_c \neq Q_h$ and $2(1-\gamma)D\big(Q_h,Q_c\big)\cos \beta < D(Q_c, Q_t)$. 
\end{theorem}

 Note the consistency comes with conditions regarding $Q_c$ and $Q_h$.  We provide geometry interpretation of such conditions in \Cref{sec:null_approx_comparison}. 

\begin{remark}
    In the later discussion, we assume $B\to \infty$, i.e., we consider the population bootstrap quantile $\hat{q}_{1-\alpha,\infty}^{m+\ell,n}$, so that we can drop $\lim_{B\to\infty}$ in \Cref{thm:partial_boots_validity} and \Cref{thm:partial_boots_consistency} for brevity.
\end{remark}

\subsubsection{Normal approximation}
\label{sec:normal_approximation}
We introduce partial bootstrap to approximate the asymptotic distribution of the test statistic $\Delta$. Since we have the form of limiting distribution of $\Delta$ in \Cref{prop:bootstrap_convergence}, it is straightforward to have an alternative approach to approximate the limiting normal distribution. We estimate the variance $\sigma_c^2$ with empirical observed samples:
\begin{align*}
\hat{\sigma}_c^2 &= \frac{(1-\gamma)^2}{m-1}\sum_{i=1}^m\Big( \hat{g}_i - \frac{1}{m}\sum_{i=1}^m\hat{g}_i\Big)^2,\\ \text{where}
\qquad\hat{g}_i &= \frac{1}{\ell}\sum_{j=1}^\ell k(x_i^c, x_j^h)
- \frac{1}{m-1}\sum_{\substack{j=1 \\ j\neq i}}^{m} k(x_i^c, x_j^c).
\end{align*}
We then take the $(1-\alpha)$-quantile, $\hat{q}_{1-\alpha}^{\prime}$, from $\cN\big(0, (1+1/c_1)\hat{\sigma}_c^2\big)$ as the critical value, and we reject $\cH_0$ if $\Delta > \hat{q}_{1-\alpha}^{\prime}$. Since we directly take the critical value by estimating the limiting distribution, the asymptotic validity and consistency of such normal approximation is straightforward. The proofs are similar to \Cref{sec:proof_partial_boots_validity} and \Cref{sec:proof_partial_boots_consistency} excluding the discussions associated with bootstrap. However, we still recommend partial bootstrap for better approximation of the null distribution, if computation cost is not a concern. Details are provided in \Cref{sec:null_approx_comparison}.

\subsubsection{Partial permutation}
\label{sec:partial_permutation}
Alternatively, one may treat $\Qhl$ as an ancillary sample entering only through the test statistic, and adapt the classical kernel two-sample test with a permutation procedure. We thus propose an alternative  algorithm named partial permutation, which replaces the step \ref{line:bootstrap} of \Cref{alg:partial_bootstrap} by a permutation step: 
\begin{enumerate}
    \item Permute the pooled observations from $\Qcm$ and $\Qtn$ to obtain $Q_{c,\pi^{b}}^m$ and $Q_{t,\pi^{b}}^n$, respectively.
    % \item Independently bootstrap $\ell$ samples from $\Qhl$ to obtain $Q_{h,b}^\ell$.
    % \item Form the fused control  $Q_{f,\pi^b}^{m+\ell} = \frac{m}{m+\ell}Q^{m}_{c,\pi^b}+ \frac{\ell}{m+\ell}Q_{h,b}^\ell$.
    \item Form the fused control  $Q_{f,\pi^b}^{m+\ell} = \frac{m}{m+\ell}Q^{m}_{c,\pi^b}+ \frac{\ell}{m+\ell}Q_{h}^\ell$.
\end{enumerate}
Denote $\tnml^b\coloneqq D^2\Big(Q_{f,\pi^b}^{m+\ell},Q_{t,\pi^{b}}^n\Big) $.  
Similar to classic permutation test, we use $\tnml$ as the test statistic, which is different from $\Delta$ used in the proposed partial bootstrap.
The critical value $\hat{c}_{1-\alpha,B}^{m+\ell,n}$ is set to be 
\begin{align*}
    \hat{c}_{1-\alpha,B}^{m+\ell,n} = \inf\Big\{u:\frac{1}{B+1}\sum_{b=0}^B \mathbbm{1}\big\{\tnml^b\leq u\big\}\geq 1-\alpha\Big\},
\end{align*} where $\tnml^0 = \tnml $.

In the next result, we show the validity and consistency results of partial permutation.

\begin{theorem}[Partial permutation validity]
\label{thm:partial_perm_validity}
Suppose \Cref{ass:kernel_moment} holds under $R = Q_c, Q_h, Q_t$, and \Cref{ass:fixed_ratios} is satisfied. Under the null hypothesis $H_0: Q_c=Q_t$,
\begin{align*}
\Pr\Big\{\tnml \ge\hat{c}_{1-\alpha,B}^{m+\ell,n}\Big\} \le\alpha.
\end{align*}
% Consequently, unconditionally $\Pr\big\{\tnml \ge \hat{c}_{1-\alpha,B}^{m+\ell,n}\big\}\le \alpha$.
\end{theorem}

The proof is provided in \Cref{sec:proof_partial_perm_validity}. Note that in partial permutation, we no longer require $Q_c$ and $Q_h$ to be non-identical for the validity to hold.

Writing $\lambda \coloneqq 1/(1+c_1)$, which equals to $n/(m+n)$ under  \Cref{ass:fixed_ratios}, we obtain the consistency of partial permutation:

\begin{theorem}[Partial permutation consistency]
\label{thm:partial_perm_consistency}
Suppose \Cref{ass:kernel_moment} holds under $R = Q_c, Q_h, Q_t$, and \Cref{ass:fixed_ratios} is satisfied. When $Q_t\neq Q_c$, with the critical value $\hat{c}_{1-\alpha, B}^{m+\ell,n}$ at the test level $\alpha$, we have 
\begin{align*} 
    \lim_{n\to\infty}\Pr \Big\{\tnml>\hat{c}_{1-\alpha,B}^{m+\ell,n}\,\Big|\,\cH_1\Big\}= 1,
\end{align*}
if and only if $Q_c=Q_h$, or $Q_c \neq Q_h$ and $2D\big(Q_h,Q_c\big)\cos\beta <\Big(\frac{1}{1-\gamma}+\lambda\Big)D\big(Q_t,Q_c\big)$.
\end{theorem}
The proof can be found in \Cref{sec:proof_partial_perm_consistency}.

\begin{remark}
\label{rem:reason_ttp}
   We show that the merged samples remain valid and consistent under $\cH_0$, regardless of the difference between $Q_c$ and $Q_h$ in the partial permutation or bootstrap algorithms. However, empirically (\Cref{sec:synthetic_experiment}) and theoretically, we find that merging $Q_h$ distributions that are sufficiently close to $Q_c$ improves both test consistency and power. Therefore, we choose a small $\theta$ so that only $\Qhl$ distributions closely matching $\Qcm$ are merged, making the fusion step essential.
\end{remark} 

% \begin{remark}[Conditioning on fixed $\Qhl$]

%     Another natural alternative is to fix $\Qhl$ in the bootstrap process of partial permutation and partial bootstrap, i.e., we only bootstrap $Q_{c,b}^{m}$ and $Q_{t,b}^{n}$ from $\Qcm$, and fuse $Q_{f,b}^{m+\ell} = \frac{m}{m+\ell}Q_{c,b}^{m}+\frac{\ell}{m+\ell}\Qhl$. The $\tnmlstar$ and $\tildetstar$ thus should change with the new $Q_{f,b}$ correspondingly. This modification changes the target distribution from $(Q_t,Q_c, Q_h)$ but $(Q_t,Q_c, \Qhl)$. The special case $Q_t=Q_c=Q_h$ is not problematic, because with high probability $Q_c=Q_h\neq \Qhl$ unless is a point mass, making the partial permutation validity exact in finite samples.
%     However, as $n \to \infty$, $\ell/n  \to 0$ with fixed $\Qhl$, the proposed TTP procedure reduces to the standard test of $Q_c=Q_t$ without $\Qhl$; the validity under $Q_c=Q_t$ using partial bootstrap no longer stand, and the joint procedure validity also fails.
% \end{remark}

\subsubsection{Comparison between partial bootstrap and partial permutation}
\label{sec:null_approx_comparison}
The partial bootstrap and partial permutation procedures generate treatment- and control-arm samples under the causal null from the same empirical source distribution $\Qcm$. The former samples from $\Qcm$ to form the bootstrap samples $Q_{t,b}^n$. This ensures that, under $\cH_0$, the resampled $Q_{t,b}^n$ and the fused $Q_{f,b}^{m+\ell}$ have the correct joint distribution, so the empirical distribution of the test statistic converges to its true null distribution $Q_c=Q_t$ when $Q_c\neq Q_h$. This is the key distinction between partial bootstrap and partial permutation lies in what role 
$\Qhl$ plays in them. When $\Qcm$ and $\Qtn $ are drawn from different distributions, treating 
$\Qhl$ merely as an auxiliary sample can lead to a poor approximation of the null distribution of $\tnml$. The issue arises because permutation also implicitly enforces exchangeability between $\Qcm$  and $\Qhl$ , effectively breaking the natural distance between them. In doing so, the permutation step ignores the distributional gap that must be preserved under the null, resulting in inflated critical values and weaker consistency. By contrast, partial bootstrap maintains the correct dependence structure, leading to a more reliable approximation. A summary of the comparison is provided in \Cref{tab:compare_two_partials}.

\begin{table}[t]
    \centering
    \begin{tabular}{c|c|c}
         &  Partial Bootstrap & Partial Permutation\\
         \hline
     Validity    & Asymptotic & Finite exact \\
     \hline
     Consistency & \makecell{Requires $\xi_h=\xi_c$ or\\ $2D\big(Q_h,Q_c\big)\cos\beta < \frac{1}{1-\gamma}D(Q_t, Q_c)$} &  \makecell{Requires $\xi_h=\xi_c$ or \\$2D\big(Q_h,Q_c\big)\cos\beta < \Big(\frac{1}{1-\gamma}+\lambda \Big)D\big(Q_t,Q_c\big)$}\\
     \hline
     \makecell{Degeneracy \\($Q_t=Q_c$)} & \checkmark (if $Q_h\neq Q_c$) & \checkmark \\
     \hline
     \makecell{Null distribution\\ approximation} & \checkmark & Fail when $Q_t\neq Q_c$
    \end{tabular}
    \caption{Comparison between partial bootstrap and partial permutation.}
    \label{tab:compare_two_partials}
\end{table}

\begin{figure}[t]
    \centering
    \includegraphics[width=0.3\linewidth]{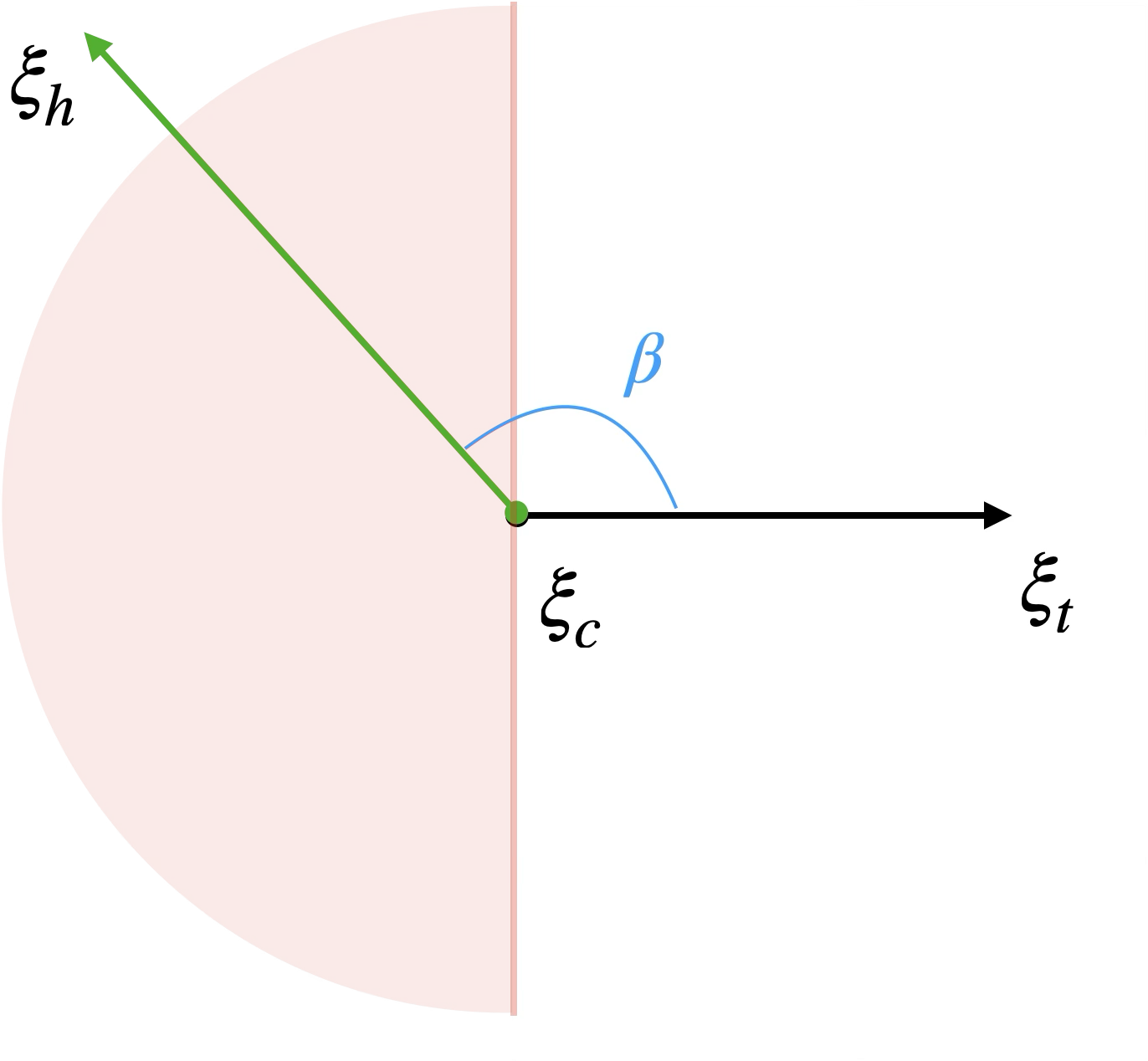}
    \caption{Geometric explanation of $\cos\beta$. The region with $\cos\beta \leq 0$ is pink.}
    \label{fig:consistency_geometry}
\end{figure}

\Cref{fig:consistency_geometry} shows a geometric interpretation of the results, which uses the definition that the kernel mean embedding maps a probability measure to a point in the RKHS $\cH_k$ \citep{song2013kernel}. In our setting, the condition for consistency depends on $D\big(Q_h,Q_c\big)$. When $\cos\beta \leq 0$, there is no restriction. However, when $\cos\beta > 0$ specifically, we require the distance 
$D\big(Q_h,Q_c\big)$ to be relatively small, which is controllable via setting the value of the radius (margin) $\theta$.

\begin{remark}[Conditional consistency]
\label{remark:conditional_consistency}
The consistency of partial permutation and partial bootstrap holds only under certain conditions involving the unknown quantities $D\big(Q_h,Q_c\big)$, $D\big(Q_t,Q_c\big)$, and $\cos \beta$. Since these quantities are not observable in practice, it is unclear whether the conditions are satisfied. A conservative check can be performed by estimating $D\big(Q_h,Q_c\big)$ and $D\big(Q_t,Q_c\big)$ and taking the largest plausible value $\cos \beta = 1$. Bounding $\cos \beta$ by $1$ in \Cref{thm:partial_boots_consistency} yields a sufficient condition for consistency: $2(1-\gamma) D(Q_h, Q_c) < D(Q_c, Q_t)$. This condition provides intuition about the family of consistent alternatives: the causality test $\cH_0$ is consistently rejected if $Q_t$ sufficiently ``further away'' (more different) from $Q_c$ than the merged distribution $Q_h$ by a factor of $2(1-\gamma)$, reflecting the cost of merging samples from potentially different distributions. In particular, when $Q_h = Q_c$, no extra bias is introduced, and the test remains consistent against all $Q_t$.
In practice, a more intuitive approach is to choose a relatively small threshold $\theta$ that upper bounds $D\big(Q_h,Q_c\big)$ for the fused control; when $\gamma$ is relatively large, a larger $\theta$ may be permissible. 
% We further analyze the case where $\theta$ decays with $n$ under fixed $m/n$ in \Cref{sec:decaying_theta}.
    
\end{remark}

\Cref{fig:test_statistic_approximation} illustrates the true null distribution of the test statistic and its approximations using either partial bootstrap or partial permutation under different data generating processes. We consider the setting where $Q_c=\cN(0,1)$. In the setting where $Q_t\neq Q_c$, $Q_t = \cN(-1,1)$. When $Q_c\neq Q_h$, we set $Q_h = \cN(2,1)$. We use sample sizes $n=600$ and $m=\ell=300$. The MMD is computed using the RBF kernel with the bandwidth chosen by the median heuristic, defined in \eqref{eqn:med_heuristic}. The middle plot of the second row highlights the weaker power of partial permutation, which produces unnecessarily large critical values when $Q_t\neq Q_c$.

The last row of \Cref{fig:test_statistic_approximation} shows the distribution obtained via the normal approximation. While the approximation is valid, it performs slightly worse than the partial bootstrap, as the asymptotic distribution becomes accurate only when sample sizes are sufficiently large. We provide more experiment results with $n$ growing in \Cref{sec:further_exp_normal_approx}. In contrast, partial bootstrap achieves a much closer approximation to the true distribution of the test statistic. This comparison highlights a key distinction: while  consistency and validity are addressed in the asymptotic behavior, in finite samples partial bootstrap can deliver substantially better power than the normal approximation.
 
\begin{figure}[t]
    \centering
    \includegraphics[width=1\linewidth]{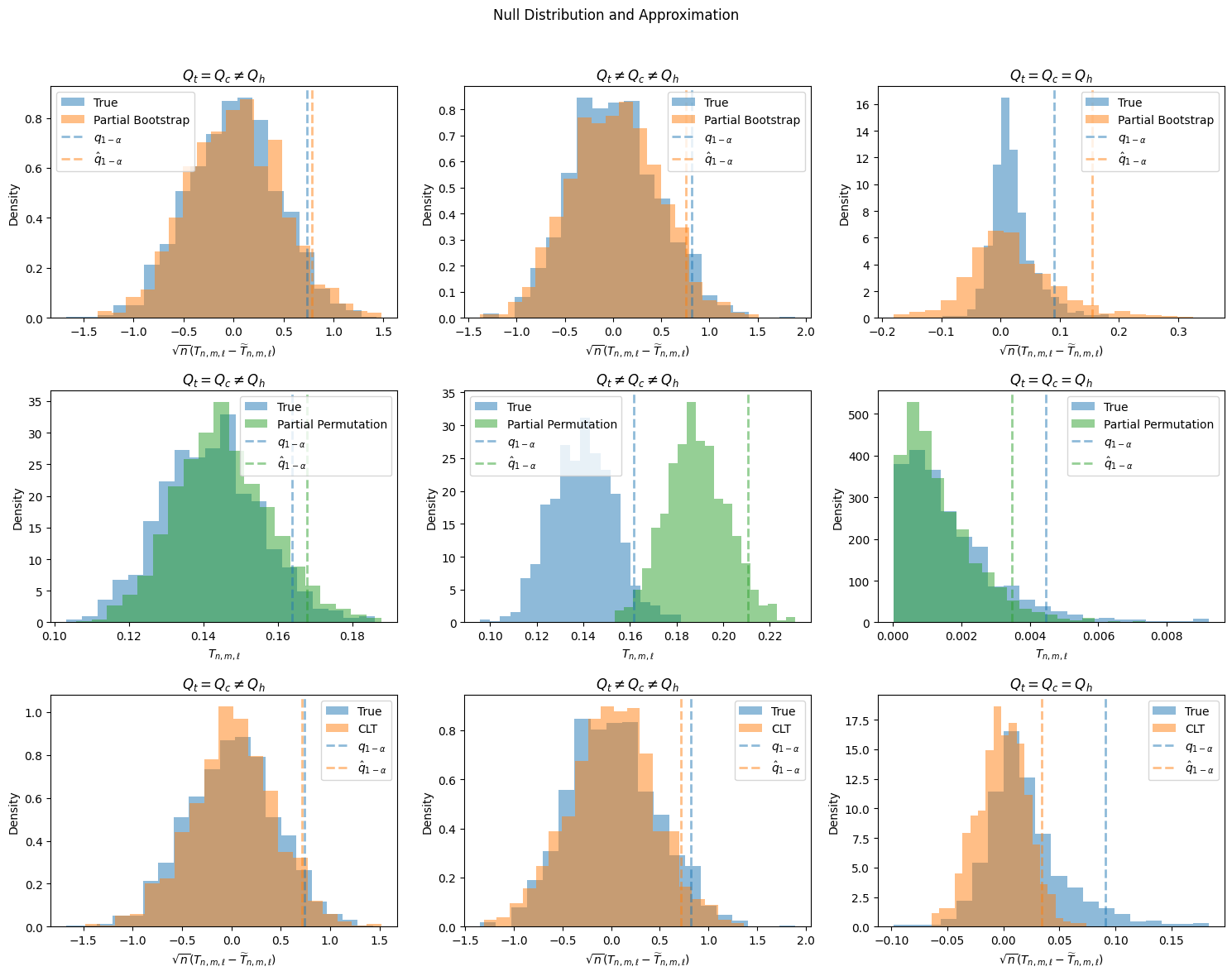}
    \caption{Comparison of the true and approximate distributions of the test statistic under partial bootstrap (first row), partial permutation (middle row) and normal approximation (last row). Results from $1000$ simulations. Within each simulation, $B=1000$ bootstraps/permutations are performed.}
    \label{fig:test_statistic_approximation}
\end{figure}

\subsection{The overall procedure}
\label{sec:overall}
Having introduced the proposed fusion and causality tests, we now show the validity of the overall TTP testing procedure.

The following result shows that the overall equivalence TTP testing procedure is well-calibrated, so long as the causality test $\phi_\ast$ and the fusion test $\phi_f$ are ``well-behaved''. Its proof is in \Cref{pf:thm:validity_general}.
\begin{theorem}
\label{thm:validity_general}
    Let $\alpha, \alpha_f \in (0, 1)$. Consider the equivalence test-then-pool procedure in \Cref{alg:equiv_ttp}. Given probability measures $Q_c, Q_h, Q_t$, assume the fusion test $\phi_f$ satisfies 
    \begin{align}
        \lim_{n \to \infty} \Pr\Big(\phi_f\big(\Qcm, \Qhl\big) = 1\Big)
        =
        \begin{cases}
            0 , &D\big(Q_h,Q_c\big) > \theta , \\
            a , &D\big(Q_h,Q_c\big) = \theta , \\
            1 , & D\big(Q_h,Q_c\big) < \theta ,
        \end{cases}
        \label{eq:fusion_test_validity} 
    \end{align}
    for some $a \in (0, \alpha_f]$. 
    Moreover, assume the causality test $\phi_\ast$ satisfies
    \begin{align}
        \limsup_{n \to \infty} \Pr\Big(\phi_\ast\big(\Qcm, \Qhl, \Qtn\big) = 1 \,\big|\, \phi_f\big(\Qcm, \Qhl\big) = 0\Big) \leq \alpha 
        , \quad \textrm{if $Q_c = Q_t $},
        \label{eq:causal_test_validity_no_merge}
    \end{align}
    and
    % \small
    \begin{align}
        \limsup_{n \to \infty} \Pr\Big(\phi_\ast\big(\Qcm, \Qhl, \Qtn\big) = 1 \,\big|\, \phi_f\big(\Qcm, \Qhl\big) = 1\Big) \leq \alpha 
        , \quad \textrm{if $Q_c = Q_t $ and $D\big(Q_h,Q_c\big) \leq \theta$}.
        \label{eq:causal_test_validity}
    \end{align}
    % \normalsize
    Then the overall equivalence TTP test satisfies 
    \begin{align*}
        \limsup_{n \to \infty} \Pr\Big(\phi_\ast\big(\Qcm, \Qhl, \Qtn) = 1 \,\big|\, \cH_0\Big) \leq \alpha .
    \end{align*}
\end{theorem}

\begin{remark}
    Equation \eqref{eq:fusion_test_validity} is a condition on the asymptotic validity and consistency of the fusion test. It is met by, for example, the MMD equivalence test reviewed in \Cref{sec:equivalence_test}, as shown in \Cref{thm:mmd_equiv_test}.
\end{remark}
\begin{remark}
    \Cref{thm:validity_general} does not require the significance level of the fusion test, $\alpha_f$, to coincide with that of the causality test, $\alpha$. This flexibility follows from the assumptions \eqref{eq:causal_test_validity_no_merge} and \eqref{eq:causal_test_validity}, which require the causality test to achieve the correct asymptotic level when applied to the unmerged and merged samples, respectively. 
\end{remark}
\begin{remark}
    Conditions \eqref{eq:causal_test_validity_no_merge} and \eqref{eq:causal_test_validity} \emph{cannot} be shown by directly applying the validity results for partial bootstrap (\Cref{thm:partial_boots_validity}) and partial permutation (\Cref{thm:partial_perm_validity}). This is because \eqref{eq:causal_test_validity_no_merge} and \eqref{eq:causal_test_validity} are \emph{conditional} on the outcome of the fusion test, whereas \Cref{thm:partial_boots_validity} and \Cref{thm:partial_perm_validity} only concern the \emph{unconditional} Type-I error, which is equivalent to always merging the samples. Nonetheless, we show that these conditions still hold for partial bootstrap and partial permutation, with proofs provided in \Cref{sec:conditional_validity}, thus establishing the validity of our proposed TTP procedure. Similarly, as we discuss below, we extend the consistency results, \Cref{thm:partial_boots_consistency,thm:partial_perm_consistency}, to the setting conditioning on the result of the fusion test.
\end{remark}

Next, we show that the proposed TTP procedure is consistent whenever the causality tests are. The proof is in \Cref{sec:proof_thm_consistency_general}.
\begin{theorem}
\label{thm:consistency_general}
     Consider the equivalence test-then-pool procedure in \Cref{alg:equiv_ttp}. Given probability measures $Q_c, Q_h, Q_t$, assume that the fusion test satisfies \eqref{eq:fusion_test_validity}, and that
     \begin{align}
         \lim_{n \to \infty} \Pr\Big(\phi_\ast\big(\Qcm, \Qhl, \Qtn\big) = 1 \,\big|\, \phi_f\big(\Qcm, \Qhl\big) = 0\Big)
         = 1 ,
         \label{eq:causal_test_consistency_no_merge}
    \end{align}
    and
    \begin{align}
         \lim_{n \to \infty} \Pr\Big(\phi_\ast\big(\Qcm, \Qhl, \Qtn\big) = 1 \,\big| \,\phi_f\big(\Qcm, \Qhl\big) = 1\Big)
         = 1, \quad \textrm{if $D\big(Q_h,Q_c\big) \leq \theta$}.
         \label{eq:causal_test_consistency}
     \end{align}
     Then the overall equivalence TTP test satisfies
     \begin{align*}
         \lim_{n \to \infty}\Pr\Big(\phi_\ast\big(\Qcm, \Qhl, \Qtn\big) = 1 \,\big|\, \cH_1\Big) = 1 .
     \end{align*}
\end{theorem}

\Cref{thm:validity_general,thm:consistency_general} work with arbitrary fusion and causality tests, and outline sufficient conditions on these components for such consistency to hold for the corresponding TTP procedure. We argue in Appendices \ref{pf:cor:validity} and \ref{sec:proof_cor_consistency} that these regularity conditions hold for the proposed MMD equivalence fusion test and the causality test with either partial bootstrap or partial permutation, which gives us \Cref{cor:validity,cor:consistency}. Notably, the significance level of the fusion test, $\alpha_f$, need \emph{not} be identical to that of the causality test, $\alpha$.

\begin{corollary}
\label{cor:validity}
    Suppose \Cref{ass:kernel_moment} holds under $R = Q_c, Q_h, Q_t$, and \Cref{ass:fixed_ratios} is satisfied. Let $\alpha, \alpha_f \in (0, 1)$, and consider the equivalence TTP test described in \Cref{alg:equiv_ttp}, where the fusion test has level $\alpha_f$. 
    Under $\cH_0: Q_c = Q_t$, 
    \begin{enumerate}
        \item Using the partial bootstrap causality test with level $\alpha$, the equivalence TTP test satisfies the following asymptotic level if $Q_h \ne Q_c$:
        \begin{align}
            \limsup_{n \to \infty}  \Pr\Big(\phi_\ast\big(\Qcm, \Qhl, \Qtn\big) = 1 \,\big|\, \cH_0\Big) \leq \alpha .
            \label{eq:overall_type_I}
        \end{align}
        \item Using the partial permutation causality test with level $\alpha$, the equivalence TTP test satisfies \eqref{eq:overall_type_I} for all $Q_h$.
    \end{enumerate}
\end{corollary}

Next, we show that the proposed TTP procedure consistently rejects alternatives whenever the causality test is capable of doing so. 
\begin{corollary}
\label{cor:consistency}
     Suppose \Cref{ass:kernel_moment} holds under $R = Q_c, Q_h, Q_t$, and \Cref{ass:fixed_ratios} is satisfied.  Let $\alpha, \alpha_f \in (0, 1)$, and consider the equivalence TTP test described in \Cref{alg:equiv_ttp}, where the fusion test has level $\alpha_f$. Under $\cH_1: Q_c \neq Q_t$, 
     \begin{enumerate}
         \item With the partial bootstrap causality test, the equivalence TTP test satisfies
         \begin{align}
             \lim_{n \to \infty}\Pr\Big(\phi_\ast\big(\Qcm, \Qhl, \Qtn\big) = 1 \,\big|\, \cH_1\Big) = 1 ,
             \label{eq:overall_power}
         \end{align}
         if $Q_c=Q_h$, or if $Q_c \neq Q_h$ and $2(1-\gamma)D\big(Q_h,Q_c\big)\cos \beta < D(Q_c, Q_t)$.
         \item With the partial permutation causality test, the equivalence TTP test satisfies \eqref{eq:overall_power} if $Q_c=Q_h$, or if $Q_c \neq Q_h$ and $2D\big(Q_h,Q_c\big)\cos\beta<\Big(\frac{1}{1-\gamma}+\lambda \Big)D\big(Q_t,Q_c\big)$.
     \end{enumerate}
\end{corollary}
\begin{remark}[Lack of full consistency]
\label{remark:decaying_theta}
    In \Cref{cor:consistency}, the conditions on the historical control arm $Q_h$ are inherited from \Cref{thm:partial_boots_consistency,thm:partial_perm_consistency}; we refer to \Cref{remark:conditional_consistency} for a discussion. In particular, when $Q_h \neq Q_c$, the proposed TTP procedure is not consistent against all alternatives $Q_t$. Intuitively, this is because, when $Q_h \neq Q_c$, merging $\Qhl$ introduces bias into the MMD statistic used in the causality test, and hence it can no longer consistently reject $Q_t$ unless the difference between $Q_t$ and $Q_c$ is substantially larger than that between $Q_h$ and $Q_c$.

    One practical mitigation is to merge with smaller probability as the sample sizes grows. This is sensible because, when the current control arm is sufficiently large, the causality test based only on the current control already has adequate power, making merging unnecessary. One way of achieving this is to let the margin $\theta$ depend on the sample size, namely taking $\theta = \theta_n \in o(1)$. In \Cref{sec:varying_theta}, we assess how varying $\theta$ affects the performance of the proposed TTP procedure. A rigorous treatment of validity and consistency for decaying sequences $\theta_n$ remains an open problem.
\end{remark}

Although both the fusion and the causality tests in our proposed TTP framework are built on MMD, the underlying principle extends naturally to other statistical metrics, such as the Kolmogorov-Smirnov (KS) metric and the Wasserstein metric. Computationally, MMD has a quadratic cost of $\operatorname{O}\!\big((n+m+\ell)^2\big)$ , whereas the KS and Wasserstein metrics require only $\operatorname{O}\!\big((m+\ell)\log(m+\ell) + n\log n\big)$. However, one major advantage of MMD is the flexibility in the kernel selection. By designing suitable kernels, the fusion and the causality tests can down-weight or up-weight certain forms of distributional discrepancy, such as mean shifts or tail behavior; see the next section for a discussion. In contrast, the standard KS statistic does not offer this type of tuning and, because it is based on the supremum norm between empirical distribution functions, it is less sensitive to discrepancies in the tails or to isolated outliers \citep{bryson1974heavy,wang2014falling}, making it more restrictive in practice. We leave the discussion of using Wasserstein metric to \Cref{sec:discussion}.

\subsection{Choosing the kernel and bandwidth}
Our framework allows flexible choices of kernels and bandwidth, which can be selected according to specific objectives. In general, we recommend employing the RBF (radial basis function) kernel $k(x,y) = \exp(-\frac{1}{2\zeta}\|x-y\|^2)$ with the bandwidth $\zeta^2$ set according to the median heuristic \citep{gretton2012kernel}, defined as
\begin{align}
\label{eqn:med_heuristic}\zeta^2_{\mathrm{med}} \coloneqq \mathrm{Median}\big\{ 
    \| v - v'\|_2^2: v, v' \in \mathcal{D}, v \neq v' 
    \big\},
\end{align}
where $\mathcal{D} = \Dnml$ if the historical sample is merged, and $\mathcal{D} = \Qcm \cup \Qtn$. 

It is also possible to tailor the kernel to target particular aspects of distributions. For example, if one is only interested in testing equality of the first moment, namely the ATE, then a linear kernel $k(x, y) = x^\top y$ can be used. Moreover, heavy-tailed or even unbounded kernels, such as the Inverse Multi-Quadric kernel \citep[Example 3]{sriperumbudur2011universality} or polynomial kernels, are more sensitive to differences in the tails \citep[S.13.2]{sadhanala2019higherorder}. However, linear and polynomial kernels are not characteristic: for example, a linear kernel cannot distinguish difference in variances \citep[Example 2]{sriperumbudur2010hilbert}. One possible solution is to combine such a kernel with a characteristic kernel, such as the RBF, and use $k(x, y) = x^\top y + \epsilon \exp(-\frac{1}{2\zeta}\|x-y\|^2)$ for some $\epsilon > 0$. This kernel remains characteristic and thus ensure that the MMD can distinguish all distributions, while still up-weighting first-moment differences. 

In our experiments, we find that bounded kernels, such as the RBF kernel, typically yields better performance, whereas unbounded ones such as the linear kernel often require large sample sizes for the asymptotic properties to hold accurately, thus leading to poor Type-I error rate control with small samples.

\subsection{Choosing the equivalence radius}
\label{sec:choose_theta}
In our setup, $\theta$ controls the tolerance for differences between the historical and current data. A larger $\theta$ allows greater heterogeneity, leading to a higher merging rate. Although the choice of $\theta$ does not affect the validity of the proposed testing procedure, it influences consistency. As discussed in \Cref{sec:varying_theta}, a large $\theta$ may reduce power when $Q_c$, $Q_h$, and $Q_t$ satisfy certain geometric conditions. Therefore, we recommend choosing a small $\theta$ to ensure a power gain. Nevertheless, the key advantage of our method lies in its robustness: it guarantees Type-I error control regardless of how different the merged and original datasets are.

Note that $\theta$ is a hyperparameter that must be specified a priori. In \Cref{remark:decaying_theta}, we discuss a potential extension of choosing decaying $\theta_n$ as a function of sample size $n$. Another alternative is to adjust $\theta$ after observing the test results, but doing so  may violate Type-I error control for $\cH_0^f$ and, consequently, inflate the Type-I error for $\cH_0$. A promising line of future work is to explore adaptive or sequential updating of $\theta$ as the function of observed data, for example, via e-values \citep{evalue}.

\section{Synthetic experiment}
\label{sec:synthetic_experiment}
We conduct synthetic experiments to compare the proposed equivalence-based TTP with the classic TTP for detecting DTE, as defined in \Cref{sec:ttp}. The experiments are designed mainly to demonstrate the validity of our method and the power improvement achieved through data fusion.

We set both $\alpha_f$ and $\alpha$ to be $0.05$. All MMD metrics are computed using the RBF kernel, where the bandwidth is selected by the median heuristic as in \eqref{eqn:med_heuristic}. We present results from $1000$ simulations with $B=1000$ bootstrap or permutation replicates per simulation.

\subsection{Equivalence TTP controls Type-I error and improves power}
\begin{figure}[!t]
    \centering
    \includegraphics[width=1\linewidth]{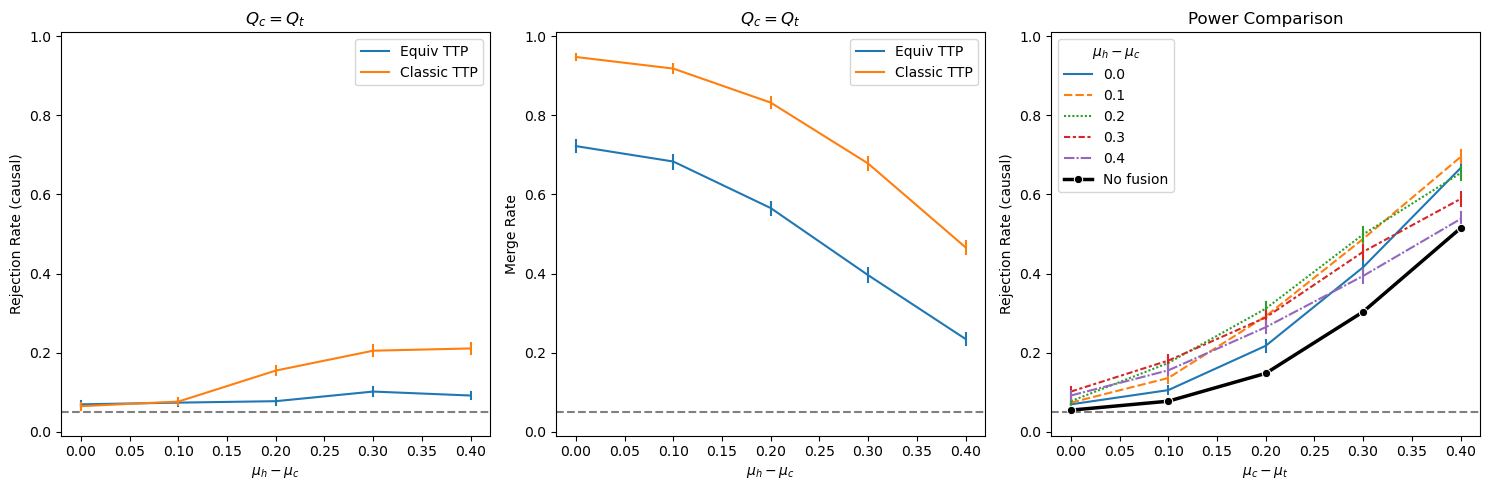}
    \caption{Type-I error (leftmost panel) and power (rightmost panel) under distributional discrepancy induced by mean shifts. The middle panel shows the proportion of the $1000$ simulations in which historical controls were merged.  The parameter $\theta$ is fixed at $0.4$. The RBF kernel and partial bootstrap procedure are used.}
    \label{fig:synthetic_mean_shift}
\end{figure}
% \begin{figure}[!t]
%     \centering
% \includegraphics[width=1\linewidth]{bandwidth_med/synthetic_perm.png}
%     \caption{$\theta$ is fixed to be $0.4$, RBF kernel, mean shift, partial permutation.}
%     \label{fig:synthetic_mean_shift_perm}
% \end{figure}

\begin{figure}[!t]
    \centering
    \includegraphics[width=1\linewidth]{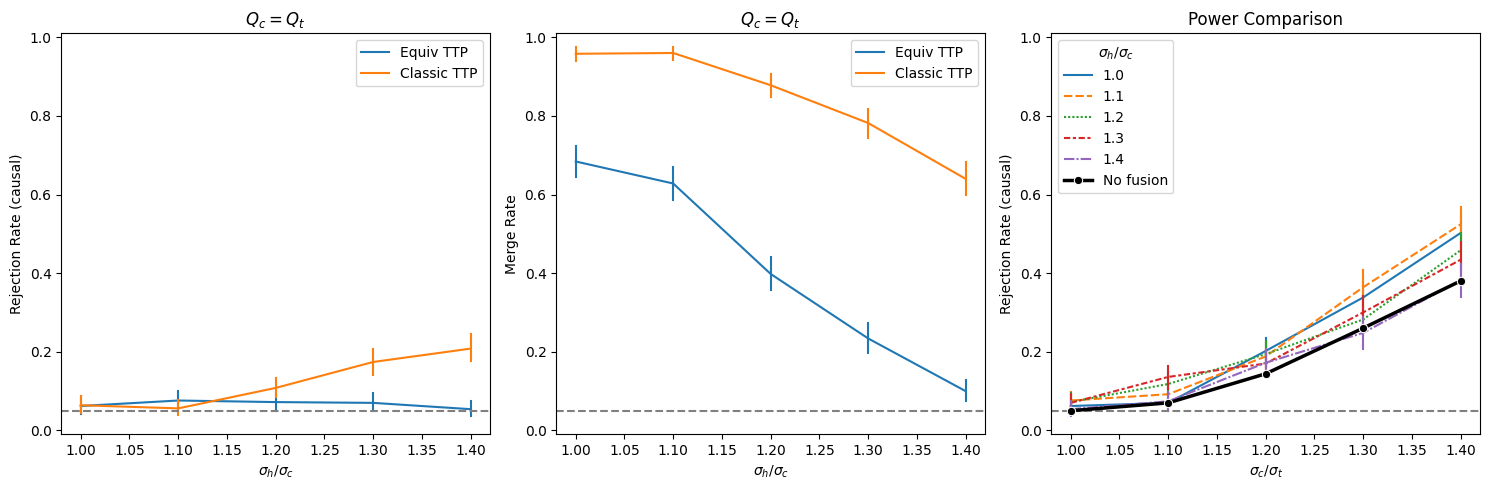}
    \caption{Type-I error (leftmost panel) and power (rightmost panel) under distributional discrepancy induced by variance shifts. The middle panel shows the proportion of the $1000$ simulations in which historical controls were merged. Type-I error and power results when distribution discrepancy is produced by variance shifts.  $\theta$ is fixed as $0.4$. The RBF kernel and partial bootstrap procedure are used.}
    \label{fig:synthetic_var_shift}
\end{figure}
\Cref{fig:synthetic_mean_shift} and \Cref{fig:synthetic_var_shift} compare the performance of the proposed equivalence TTP, the classic TTP, and a test without fusing historical controls under mean-shift and variance-shift settings, respectively. In these experiments, we set $n=100$, $m=50$ and $\ell=100$.  

For the mean-shift scenario (\Cref{fig:synthetic_mean_shift}), we generate data with $Q_t = \cN(0,1)$ with $\mu_t=0$, $Q_c = \cN(\mu_c-\mu_t,1)$, and $Q_h = \cN(\mu_t-\mu_c,1)$, with $\mu_c$, $\mu_h$ varying. For the variance-shift scenario (\Cref{fig:synthetic_var_shift}), we use  $Q_t = \cN(0,1)$ (thus $\sigma_t^2=1$), $Q_c = \cN(0,\sigma_c^2/\sigma_t^2)$, and $Q_h = \cN(0,\sigma_h^2/\sigma_c^2)$. We vary $\sigma_h^2$ and $\sigma_c^2$ as shown in the experiment results.  

The results demonstrate the effectiveness of the proposed equivalence TTP method. Using $\theta=0.4$, the procedure successfully fuses historical controls selectively, as shown in the middle plots, while maintaining control of the Type-I error rate. This highlights one of the main advantages of our proposed method, as it differs from the classic TTP which exhibits clear Type-I error inflation (left panels). At the same time, our method achieves substantial power gains relative to not fusing controls, confirming the benefits of incorporating historical data when it is sufficiently similar to the current control group.

\subsection{Effects of varying $\theta$}
\label{sec:varying_theta}
We also investigate how varying the choice of
 $\theta$ affects both the merge rate and the power of the causality test under different ``directions'' of distribution shifts. We follow the same experiment set up used in \Cref{fig:synthetic_mean_shift}. As shown in \Cref{fig:type1_theta_vary}, larger values of $\theta$ increase the probability of merging even when $Q_h$ and $Q_c$ differ more substantially. Nevertheless, the Type-I error of the causality test remains close to the nominal significance level $\alpha=0.05$, indicating that validity is preserved across a range of thresholds.

\Cref{fig:power_theta_vary} examines the effect of
 $\theta$ on statistical power. Here we allow $\mu_c-\mu_t$ to take negative values. Since $\mu_h-\mu_c>0$, this setting corresponds to $Q_t$ and $Q_h$ moving in the same direction away from $Q_c$, i.e., $\cos\beta>0$. In this case, merging $Q_h$ reduces the apparent difference between  $Q_t$ and $Q_c$ which can lead to reduced power. Indeed, for sufficiently large $\theta$, we observe power dropping below the no-fusion baseline due to this cancellation effect. We highlight this result as a cautionary note: overly permissive choices of $\theta$, can harm power, reinforcing the importance of choosing $\theta$ conservatively in practice.

% \begin{figure}[!t]
%   \begin{subfigure}[b]{0.45\textwidth}
%           \centering    \includegraphics[width=0.8\linewidth]{bandwidth_med/type1_theta_vary.png}
%     \caption{$Q_t = Q_c = \cN(0,1)$, $Q_h = \cN(\mu_t-\mu_c,1)$, with $\theta$ varying; partial bootstrap.}
%     \label{fig:type1_theta_vary}
%   \end{subfigure}
%   \hfill
%   \begin{subfigure}[b]{0.45\textwidth}
%     \centering    \includegraphics[width=0.8\linewidth]{bandwidth_med/perm_type1_theta_vary.png}
%     \caption{$Q_t = Q_c = \cN(0,1)$, $Q_h = \cN(\mu_t-\mu_c,1)$, with $\theta$ varying; partial permutation.}
%     \label{fig:perm_type1_theta_vary}
%   \end{subfigure}
%   \caption{Comparison with varying $\theta$. (a) Partial bootstrap; (b) Partial permutation.}
%   \label{fig:synthetic_combined}
% \end{figure}

\begin{figure}[!t]
    \centering    \includegraphics[width=0.8\linewidth]{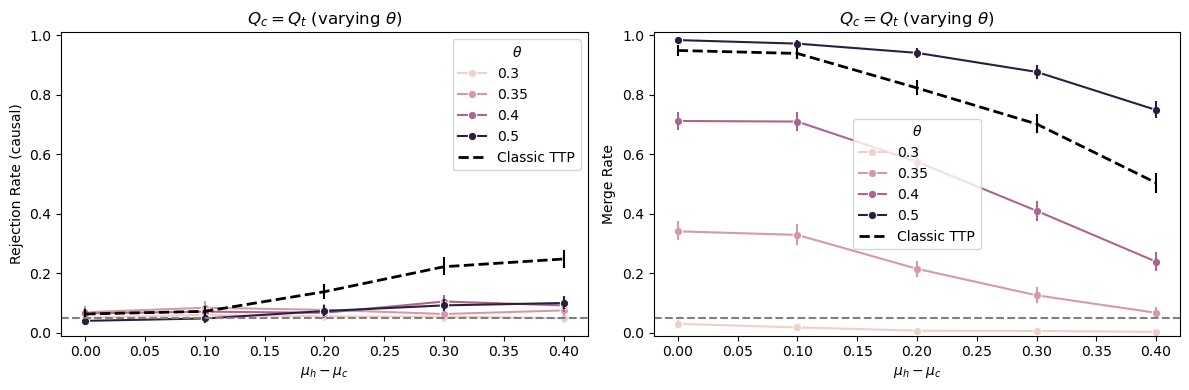}
    \caption{Type-I error (left panel) is controlled across varying $\theta$, despite differing merge rates (right panel). In these experiments, $Q_t = Q_c = \cN(0,1)$, $Q_h = \cN(\mu_h-\mu_c,1)$. The RBF kernel is used with the partial bootstrap procedure.}
    \label{fig:type1_theta_vary}
\end{figure}
\begin{figure}[!t]
    \centering    \includegraphics[width=0.8\linewidth]{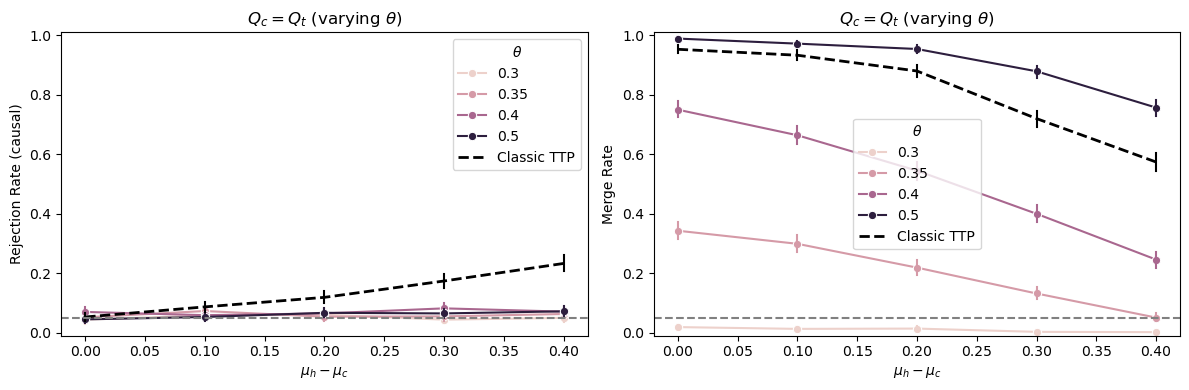}
    \caption{Type-I error (left panel) is controlled across varying $\theta$, despite differing merge rates (right panel). In these experiments, $Q_t = Q_c = \cN(0,1)$, $Q_h = \cN(\mu_h-\mu_c,1)$. The RBF kernel is used with the partial permutation procedure.}
    \label{fig:perm_type1_theta_vary}
\end{figure}

\subsection{Comparison between partial bootstrap and partial permutation}
In Figures 
% \ref{fig:synthetic_mean_shift_perm}, 
\ref{fig:perm_type1_theta_vary} and  \ref{fig:perm_power_theta_vary},  we present results from partial permutation under the same data generation mechanism with mean shift, as in \Cref{fig:type1_theta_vary}. As expected, the Type-I error and the power follow similar patterns to those in the partial bootstrap setting. Meanwhile, as we suggest above, since partial permutation fails to approximate the null distribution when $Q_c\neq Q_t$, the power from partial permutation is indeed smaller than what we have in partial bootstrap---we have the same observation in results of  \Cref{tab:vary_kernel_bandwidth}.

\begin{figure}[!t]
    \centering    \includegraphics[width=1\linewidth]{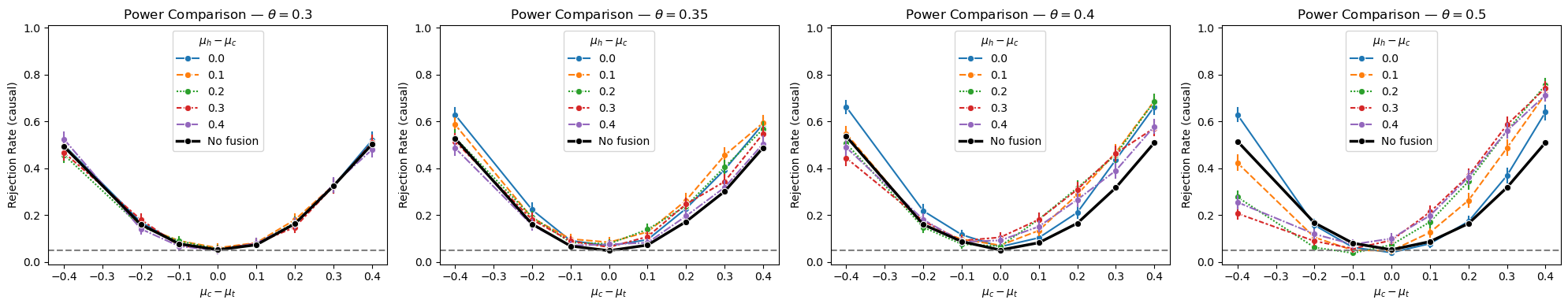}
    \caption{Power comparison with $\theta$ varying. In these experiments, $Q_t = \cN(0,1)$ with $\mu_t=0$, $Q_c = \cN(\mu_c-\mu_t,1)$, and $Q_h = \cN(\mu_h-\mu_c,1)$. The RBF kernel and partial bootstrap are used.}
    \label{fig:power_theta_vary}
\end{figure}

\begin{figure}[!t]
    \centering    \includegraphics[width=1\linewidth]{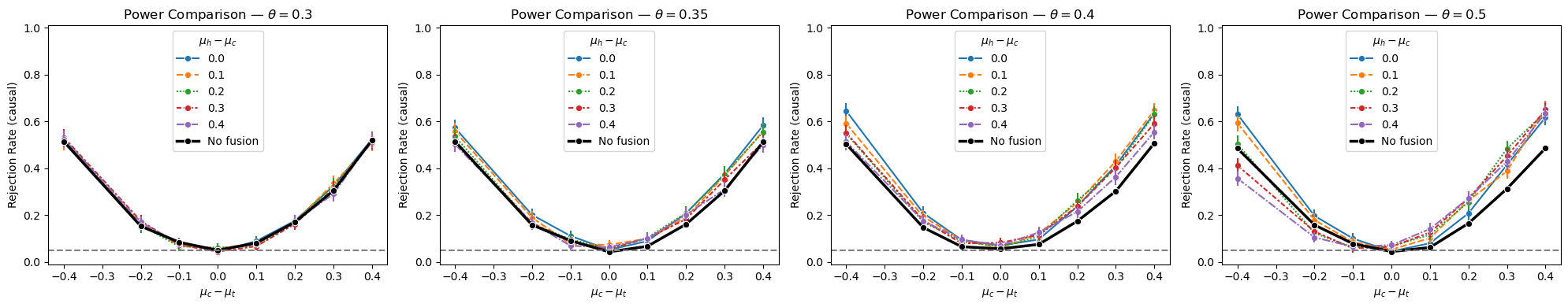}
    \caption{Power comparison with $\theta$ varying. In these experiments, $Q_t = \cN(0,1)$ with $\mu_t=0$, $Q_c = \cN(\mu_c-\mu_t,1)$, and $Q_h = \cN(\mu_h-\mu_c,1)$. The RBF kernel and partial permutation are used.}
    \label{fig:perm_power_theta_vary}
\end{figure}

\subsection{Comparison of varying kernel bandwidth}
We evaluate the performance of the proposed equivalence TTP using partial bootstrap and partial permutation under varying kernel bandwidths of the selected RBF kernel (\Cref{tab:vary_kernel_bandwidth}). For fixed $\theta$, the kernel bandwidth influences the merge rate. Nevertheless, equivalence TTP consistently controls the Type-I error. Moreover, the partial bootstrap method generally achieves higher power than partial permutation, in agreement with our earlier discussion.
\begin{table}[!ht]
\centering
\begin{tabular}{c cc cc cc}
\toprule
& \multicolumn{2}{c}{$\mu_c-\mu_t=-0.4$}
& \multicolumn{2}{c}{$\mu_c=\mu_t$}
& \multicolumn{2}{c}{$\mu_c-\mu_t=0.4$} \\
\cmidrule(lr){2-3}\cmidrule(lr){4-5}\cmidrule(lr){6-7}
\makecell{\textbf{Bandwidth}\\ \textbf{in RBF Kernel}} & \textbf{Merge} & \textbf{Causal} & \textbf{Merge} & \textbf{Causal} & \textbf{Merge} & \textbf{Causal} \\
\midrule
Median Heuristic & 0.285 & 0.577 (0.513) & 0.307 & 0.064 (0.051) & 0.272 & 0.572 (0.535) \\
0.5  & 0.319 & 0.528 (0.511) & 0.342 & 0.070 (0.063) & 0.354 & 0.536 (0.566) \\
1    & 0.673 & 0.572 (0.559) & 0.638 & 0.074 (0.053) & 0.694 & 0.711 (0.629) \\
1.5  & 0.824 & 0.581 (0.566) & 0.807 & 0.074 (0.065) & 0.832 & 0.739 (0.677) \\
\bottomrule
\end{tabular}
\caption{Merge rates and causality test rejection rates by bandwidth and $\mu_c-\mu_t$ blocks by partial bootstrap, averaged over 1000 simulations, with $1000$ bootstraps / permutations in each simulation. $\alpha_f=\alpha = 0.05$, $\theta$ is fixed at $0.4$, $\mu_h-\mu_c=0.2$. For the causality test, rejection rates from the partial bootstrap are reported directly, with the corresponding partial permutation results given in parentheses. All experiments are done with RBF kernels.}
\label{tab:vary_kernel_bandwidth}
\end{table}

\subsection{Comparison with linear kernel}
As discussed earlier, our method accommodates various kernel choices, which should be selected according to the task. The linear kernel effectively captures differences in first moments but may fail to detect higher-order discrepancies such as variance shifts. In \Cref{fig:synthetic_var_shift_linear}, we apply the linear kernel to the variance shift setting from \Cref{fig:synthetic_var_shift}. Although the underlying distributions differ in variance, the MMD two-sample tests fail to detect this difference, regardless of whether data fusion is applied.
\begin{figure}[!t]
    \centering
    \includegraphics[width=1\linewidth]{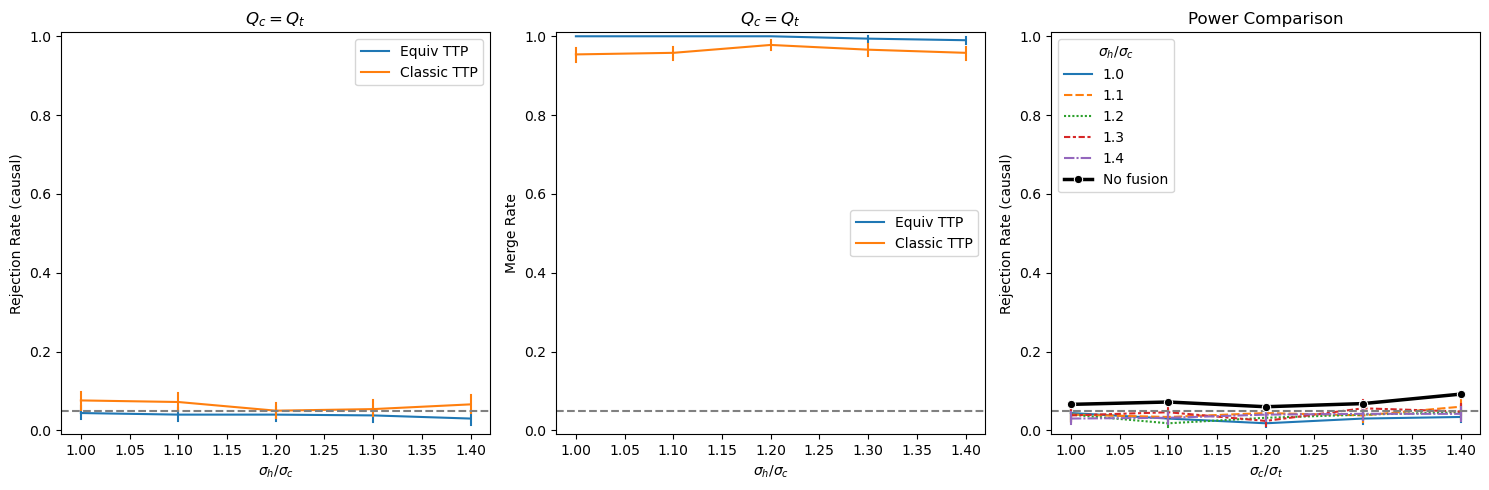}
    \caption{Experiment results when distributional discrepancy produced by shifting variance. $\theta$ is fixed at $1$. The linear kernel and partial bootstrap are used.}
    \label{fig:synthetic_var_shift_linear}
\end{figure}

\section{Prospera program}
\label{sec:prospera}
In this section, we apply the equivalence TTP procedure on the \emph{Prospera} program. Originally launched in 1997 as Progresa, Mexico’s conditional cash transfer (CCT) program provides cash to poor households conditional on school attendance, health visits, and participation in nutrition programs \citep{skoufias2005progresa}. The initial evaluation was conducted as a randomized controlled trial across 505 villages in seven poor states: 320 villages were assigned to treatment (eligible poor households received stipends conditional on compliance) and 185 to control (no program). We use village-level school enrollment rates from 1998 as post-treatment outcomes, with the 1997 baseline serving as pre-treatment information.

We first examine whether the program improved school enrollment. A standard two-sample test for average treatment effects, applied to the complete 1998 data, rejects the null hypothesis $\cH_0$ of equal enrollment rate distributions between treatment and control villages at the $\alpha = 0.05$ level. This confirms that Prospera significantly increased school enrollment, consistent with prior findings on the program’s effectiveness.

To evaluate performance in settings with limited contemporary controls, we randomly subsample $50$ control villages from the 1998 control group and use their corresponding 1997 baseline data as historical controls. We also randomly sample $200$ villages from the treatment group for testing. \Cref{fig:prospera} displays the enrollment rate distribution of treatment, control, and historical groups in one such simulation. We repeat this procedure for $1000$ simulations in total.

\begin{figure}[!t]
    \centering
    \includegraphics[width=1\linewidth]{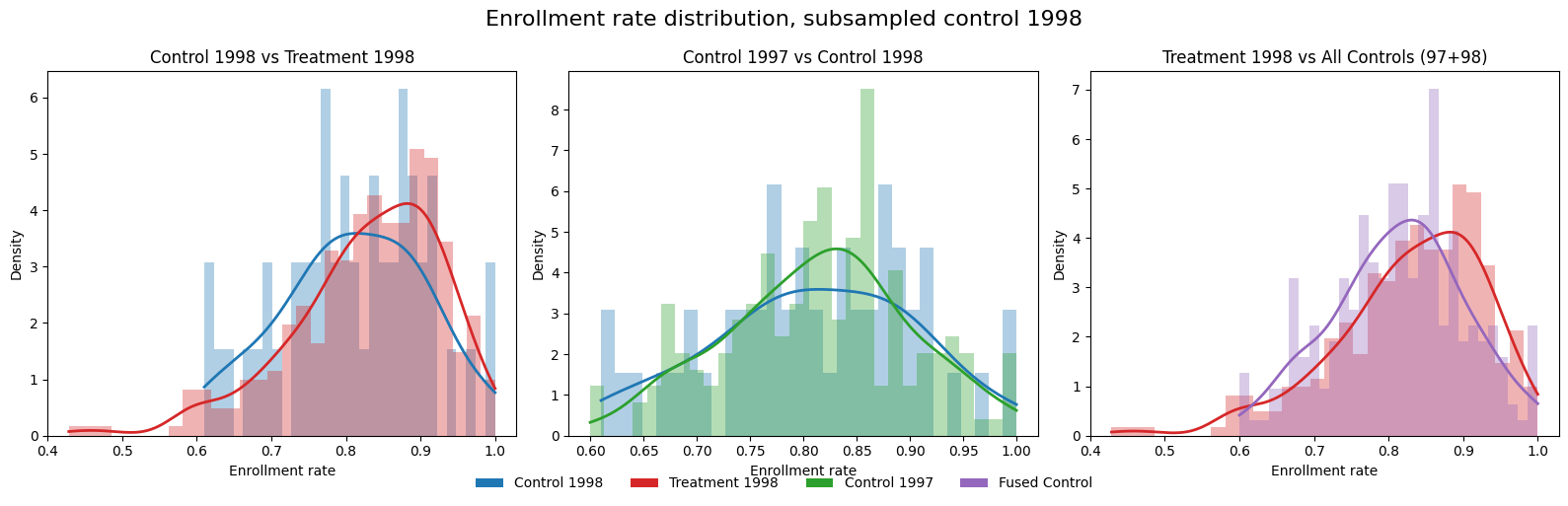}
    \caption{Distribution of enrollment rates across groups.}
    \label{fig:prospera}
\end{figure}

  We set the fusion test level as $\alpha_f = 0.05$ and and the significance level of the causality test to $\alpha = 0.05$. As before, the kernel bandwidth $\zeta$ is chosen by the median heuristic \eqref{eqn:med_heuristic} for the RBF kernel. Under these settings, the proposed equivalence TTP procedure accepts fusing the current control group with historical control group (1997) at a threshold of $\theta = 0.5$ in $99.4\%$ of the 1000 replications. 
  
  \Cref{tab:rejection_rates} reports the rejection rates of the causality test against $\cH_0$ across different methods. The proposed TTP achieves the highest power, outperforming a classical distribution test that does not fuse control groups. Furthermore, compared with their mean-based counterparts, distributional tests demonstrate superior power, as the differences in  distributions additional to the mean are captured. Together, these findings validate the effectiveness of the proposed equivalence TTP framework.

\begin{table}[!t]
\centering
% \begin{minipage}{0.5\textwidth}
%     \centering
    \begin{tabular}{cc}
        \toprule
        Method & Rejection Rates \\
        \midrule
        Equivalence TTP & 0.61 \\
        Distribution test (without fusion) & 0.40  \\
        Mean test (with fusion) & 0.36\\
        Mean test (without fusion) & 0.23 \\
        \bottomrule

    \end{tabular}
    \caption{Rates of rejecting $\cH_0$ of $1000$ simulations.}
    \label{tab:rejection_rates}
% \end{minipage}\hfill
% \begin{minipage}{0.4\textwidth}
%     \centering
%     \begin{figure}[H]
%         \centering
%         \includegraphics[width=1\linewidth]{Prospera_mmd.png}
%         \caption{Distribution of MMD.}
%         \label{fig:prospera_mmd}
%     \end{figure}
% \end{minipage}
\end{table}

\section{Discussion}
\label{sec:discussion}
In this paper, we focus on the fusion of control arms (and, symmetrically, treatment arms) within randomized clinical trials, and therefore make no assumptions such as unconfoundedness. In practice, however, it is increasingly common to consider fusing real-world observational data with experimental data. Extending our framework to such settings requires additional assumptions, most notably ignorability and the comparability of populations---lack of ignorability in the observational data may induce confounding bias; even under ignorability, substantial differences in covariate distributions between datasets may alter the target estimand.

Our proposed method can be adapted to address these concerns. For instance, MMD-based tests can be applied to assess distributional equivalence of covariates between observational and experimental samples, providing a principled diagnostic for population alignment in high-dimensional settings. Similarly, when imbalance is detected, our framework could be combined with covariate adjustment or weighting strategies to mitigate bias before fusion. In this way, the same equivalence-testing tools that validate control-arm pooling in trials can also serve as safeguards when extending pooling to broader causal inference contexts.

The choice of radius $\theta$ presents an interesting direction for future work. As discussed in \Cref{sec:choose_theta}, intuitively, $\theta$ should not be too large, so that the merged historical control distribution $\Qhl$ remains similar to the current control distribution $\Qcm$; however, setting $\theta$ too small may make  the fusion test overly conservative. We provide ideas on potential extensions in \Cref{sec:choose_theta}, \Cref{remark:decaying_theta} and \Cref{sec:varying_theta}.

We discuss at the end of \Cref{sec:overall} on using other distributional distance metrics, and one of the potential direction is to use the Wasserstein metric. Different from the Kolmogorov-Smirnov metric, the Wasserstein distance can in principle be made more flexible through an appropriate choice of cost function, but we do not pursue this direction here, as both the equivalence testing employed in this paper and most existing methods for distributional treatment effects are formulated in terms of MMD. Extending the proposed TTP framework to Wasserstein-based test statistics is therefore a natural direction for future work.

Another potential extension involves merging multiple data sources. Instead of relying on binary ``accept or reject'' decisions, one could assign weights based on their similarity to the current control group. Specifically, the kernel-based distance $D(Q_c,Q_h)$ from  the kernel equivalence test could be used to construct a smooth weighting function, allowing closer sources to contribute more heavily to the fused control distribution while down weighting or excluding more distant sources. Combining information from several candidate sources in this manner would require developing procedures for joint calibration.

Finally, while this work is motivated by causal inference, the proposed framework has potential applications beyond causal settings, such as domain adaptation, transfer learning, and robust statistical testing in machine learning contexts.

\subsection*{Acknowledgements}

The authors thank Wenkai Xu for helpful discussions. Linying Yang is supported by the EPSRC Centre for Doctoral Training in Modern Statistics and Statistical Machine Learning (EP/S023151/1) and Novartis.  Robin J.~Evans has been supported by Novo Nordisk. 

% For the purpose of open access, the authors have applied a CC-BY public copyright licence to any author accepted manuscript (AAM) arising from this submission.

\eject

\appendix
\section{Auxiliary results}
This section collects auxiliary results that support the main theorems.
\subsection{Consequences of a Fixed Pooling Ratio}
The pooled sample $\Qfml$ is formed by combining two mutually independent samples $\Qcm$ and $\Qhl$, which are i.i.d.\ samples from $Q_c$ and $Q_h$, respectively. However, the pooled sample $\Qfml$ is \emph{not} an i.i.d.\ sample from the mixture distribution $\frac{m}{m+\ell}Q_c + \frac{\ell}{m+\ell}Q_h$. This is because the sample sizes $m$ and $l$ are \emph{deterministic}, whereas under i.i.d.\ sampling from the mixture, the number of observations drawn from each component would be random and follow a binomial distribution. In fact, $\Qfml$ is \emph{exchangeable}, meaning that the joint distribution of $\Qfml$ remains unchanged under permutations. 

Despite this distinction, this technical subtlety does not significantly complicate our analysis. For example, the kernel mean embedding of $\Qfml$ coincides exactly with that of an i.i.d.\ sample from the mixture $\frac{m}{m+\ell}Q_c + \frac{\ell}{m+\ell}Q_h$, as shown in the following result. 

\begin{lemma}
\label{lem:merged_sample_kme}
    Let $\{X_i\}_{i=1}^\infty$ and $\{Y_i\}_{i=1}^\infty$ be sequences of independent random variables drawn from $Q_1, Q_2$. For any positive integers $n_1, n_2$, define $n = n_1 + n_2$, and let $Q_n$ be the empirical measure based on $\{Z_i\}_{i=1}^n$, where $Z_i = X_i$  if $1 \leq i \leq n_1$ and $Z_i = Y_i$ if $n_1 < i \leq n$. Moreover, let $Q_{1,n_1}$ and $Q_{2,n_2}$ be the empirical samples based on $\{X_i\}_{i=1}^{n_1}, \{Y_i\}_{i=1}^{n_2}$, respectively, and define the mixture distribution $Q = \rho Q_1 + (1-\rho) Q_2$. Then the kernel mean embedding of $Q_n$ takes the form
    \begin{align*}
        \hat{\xi}_{Q_n}
        = \frac{n_1}{n} \E_{X \sim Q_1}k(\cdot, X) + \frac{n_2}{n} \E_{Y \sim Q_2}k(\cdot, Y)
        = \frac{n_1}{n} \hat{\xi}_{n_1} + \frac{n_2}{n} \hat{\xi}_{n_2}
        .
    \end{align*}
    % If furthermore $n_1, n_2$ are sequences of positive integers with $n_1/ n_2 \to \rho \in [0, 1]$, then $\hat{\xi}_{Q_n} \to $
\end{lemma}
Since $\hat{\xi}_{Q_n}$ can be decomposed into the kernel mean embeddings based on $Q_{1,n_1}$ and $Q_{2,n_2}$, and since $Q_{1,n_1}$ and $Q_{1,_n2}$ are mutually independent, the asymptotic properties of $\hat{\xi}_{Q_n}$ can be readily derived from classic results on kernel mean embeddings of independent samples. We will leverage this observations in most of our proofs, such as \Cref{sec:proof_bootstrap_convergence}.

\begin{proof}[Proof of \Cref{lem:merged_sample_kme}]
    This follows immediately from the definition of kernel mean embeddings of empirical samples:
    \begin{align*}
        \hat{\xi}_{Q_n}
        = \frac{1}{n}\sum_{i=1}^n Z_i
        &= \frac{1}{n}\sum_{i=1}^{n_1} X_i + \frac{1}{n}\sum_{i=n_1+1}^{n} Y_i
        = \frac{n_1}{n} \E_{X \sim Q_{1,n_1}}k(\cdot, X) + \frac{n_2}{n} \E_{Y \sim Q_{2,n_2}}k(\cdot, Y)
        \\
        &= \frac{n_1}{n} \hat{\xi}_{n_1} + \frac{n_2}{n} \hat{\xi}_{n_2}.
    \end{align*}
\end{proof}

\subsection{Validity with merged samples}
\label{sec:conditional_validity}
This section contains some auxiliary results that will be used to show the validity of the overall procedure in \Cref{cor:validity}; to be more specific, we show that \eqref{eq:causal_test_validity} is met.  These results are extensions of \Cref{thm:partial_boots_validity} and \Cref{thm:partial_perm_validity} to the case where we condition on the event that the fusion test, $\cH^f_0$, is rejected, namely that historical and current controls are merged. We present the results for both partial bootstrap (\Cref{app:partial_boot_validity_conditional}) and partial permutation (\Cref{app:partial_perm_validity_conditional}).

\subsubsection{Partial bootstrap}
\label{app:partial_boot_validity_conditional}

We first show that the same weak limit in \Cref{prop:bootstrap_convergence} holds conditional on the event that the fusion test is rejected.

\begin{proposition}[Weak convergence of the partial bootstrap provided merged controls under non-identical control groups]
\label{prop:bootstrap_convergence_conditional} 
Assume \Cref{ass:kernel_moment} holds under $R = Q_c, Q_h, Q_t$, and \Cref{ass:fixed_ratios} is satisfied. Let $\cA_n = \big\{\phi_f(\Qcm, \Qhl) = 1\big\}$ denote the event of rejecting the fusion null, and assume $\tau \coloneqq\lim_{n\to\infty}\Pr(\cA_n) > 0$. Under the causal null $\cH_0:Q_c=Q_t$, the following holds conditionally on $\cA_n$:
\begin{align*}
   \Delta \rightsquigarrow_{\Pr} \cN\big(0,4(1+1/c_1)\sigma_c^2\big)\quad \text{and} \quad \Delta^\ast \rightsquigarrow_{\Pr} \cN\big(0,4(1+1/c_1)\sigma_c^2\big) 
   ,
\end{align*}
where \begin{align*}
    \sigma_c^2 &\coloneqq\Var_{X\sim Q_c}\Big[\E_{Y\sim Q_f}k(X, Y) - \E_{Z\sim Q_c}k(X, Z)\Big].
\end{align*}
We thus obtain
\begin{align*}
    \sup_{z\in\R} \Big|{\Pr}^\ast\big\{\Delta^\ast\leq z\big\} - \Pr\big\{\Delta\leq z\big\} \Big|\overset{a.s.}{\longrightarrow}0.
\end{align*}
\end{proposition}

\begin{proof}
    Most parts in the proof of \Cref{prop:bootstrap_convergence} (present in \Cref{sec:proof_bootstrap_convergence}) remain valid, except the proof of $A_n-A = \Opr(n^{-1/2})$. We can not we cannot apply the standard CLT in Hilbert space directly on $A_n = (1-\gamma)\big(\hat{\xi}_{h,n}-\hat{\xi}_{c,m}\big)$, as it is no longer weighted average of independent samples when conditioning on the event $\cA_n = \big\{\phi_f(\Qcm, \Qhl) = 1\big\}$.  Nevertheless, we show that this claim still holds, so long as $\lim_{n\to\infty}\Pr(\cA_n) > 0$.

    We aim to show that $A_n-A = \Opr(n^{-1/2})$ holds under the sequence of conditional probabilities $\Pr( \cdot \;|\; \cA_n)$, namely there exists a constant $c$ such that
    \begin{align*}
        \Pr\Big(\big\|A_n - A\big\| > c n^{-\frac{1}{2}} \,\big|\, \cA_n\Big) \to 0 .
    \end{align*}
    Using the CLT argument in \Cref{prop:bootstrap_convergence}, we know that this holds \emph{unconditionally}. That is, we can find a constant $c$ such that
    \begin{align}
        \Pr\Big(\big\|A_n - A\big\| > c n^{-\frac{1}{2}}\Big) \to 0 .
        \label{eq:An_unconditional_convergence}
    \end{align}
    Define $\tau \coloneqq \lim\nolimits_{n\to\infty} \Pr\big(\cA_n\big)$, which exists and is positive by assumption. In particular,  there exists $N\in\mathbb{N}_+$ such that $\Pr(\cA_n)>0$ for all $n\ge N$. In what follows, we restrict attention to  $n\ge N$, so that $\Pr( \cdot \;|\; \cA_n)$ is well-defined. For such $n$,
     \begin{align*}
        \Pr\Big(\big\|A_n - A\big\| > c n^{-\frac{1}{2}} \,\big|\, \cA_n\Big)
        &= \frac{\Pr\Big(\big\{\big\|A_n - A\big\| > c n^{-\frac{1}{2}}\big\}\cap\cA_n\Big)}{\Pr\big(\cA_n\big)}\\
        &\le\frac{\Pr\Big(\big\|A_n - A\big\| > c n^{-\frac{1}{2}}\Big)}{\Pr\big(\cA_n\big)}\\
        &\to \frac{\lim\nolimits_{n\to\infty} \Pr\Big(\big\|A_n - A\big\| > c n^{-\frac{1}{2}}\Big)}{\tau}
        = 0
        ,
    \end{align*}
    where in the last line we have used the assumption that $\tau > 0$ and \eqref{eq:An_unconditional_convergence}. This shows that $A_n-A = \Opr(n^{-1/2})$ holds conditionally under $\Pr(\cdot \;|\; \cA_n)$.
\end{proof}

Using this result, we can then show a conditional version of the partial bootstrap validity.

\begin{theorem}[Partial bootstrap validity provided merged controls under non-identical control groups]
\label{thm:partial_boots_validity_conditional}
    Suppose \Cref{ass:kernel_moment} holds under $R = Q_c, Q_h, Q_t$, and \Cref{ass:fixed_ratios} is satisfied. Let $\hat{q}_{1-\alpha,B}^{m+\ell,n}$ be the partial bootstrap critical value defined in \eqref{eq:partial_boot_quantile}. Moreover, let $\cA_n = \big\{\phi_f(\Qcm, \Qhl) = 1\big\}$ denote the event of rejecting the fusion null, and assume $\tau \coloneqq\lim_{n\to\infty}\Pr(\cA_n) > 0$. Under the causal null hypothesis $\cH_0:Q_c=Q_t$ and assuming $D(Q_c, Q_h) \leq \theta$, the test, conditionally on merging, has asymptotic level $\alpha$:
\begin{align*}
    \limsup_{n\to\infty}\lim_{B\to\infty}\Pr \Big\{\Delta>\hat{q}_{1-\alpha,B}^{m+\ell,n}\,\Big|\,\cH_0, \cA_n\Big\}\leq\alpha.
\end{align*}
\end{theorem}

\begin{proof}
    The claim holds by following the same proof in \Cref{thm:partial_boots_validity} and using \Cref{prop:bootstrap_convergence_conditional} instead of \Cref{prop:bootstrap_convergence}.
\end{proof}

\subsubsection{Partial permutation}
\label{app:partial_perm_validity_conditional}

The following result shows that the partial permutation also remains valid conditionally on the event that the fusion test is rejected.

\begin{theorem}[Partial permutation validity provided merged controls]
\label{thm:partial_perm_validity_conditional}
Suppose \Cref{ass:kernel_moment} holds under $R = Q_c, Q_h, Q_t$, and \Cref{ass:fixed_ratios} is satisfied. Let  $\cA_n = \big\{\phi_f(\Qcm, \Qhl) = 1\big\}$ denote the probability of rejecting the fusion null, and assume $\tau \coloneqq\lim_{n\to\infty}\Pr(\cA_n) > 0$. Under the null hypothesis $H_0: Q_c=Q_t$,
\begin{align*}
\Pr\Big\{\tnml \ge\hat{c}_{1-\alpha,B}^{m+\ell,n}\,\big|\, \cA_n\Big\} \le\alpha.
\end{align*}
\end{theorem}

\begin{proof}
Since $\tau>0$, there exists $N\in\mathbb{N}_+$ such that $\Pr(\cA_n)>0$ for all $n\ge N$. In what follows, we restrict attention to  $n\ge N$, so that $\Pr( \cdot \;|\; \cA_n)$ is well defined.
Following the proof in \Cref{sec:proof_partial_perm_validity}, we have \eqref{eq:partial_perm_conditional_validity}. Given $\cA_n$ is  $\sigma\big(Z,\Qhl\big)$-measurable. We have
\begin{align*}
    \Pr\Big\{\tnml \ge\hat{c}_{1-\alpha,B}^{m+\ell,n}\,\big|\, \cA_n\Big\} &= \E\bigg\{\Pr\Big\{\tnml \ge\hat{c}_{1-\alpha,B}^{m+\ell,n}\,\big|\, \cA_n, Z,\Qhl\Big\}\Big|\cA_n\bigg\}\\
    &= \E\bigg\{\Pr\Big\{\tnml \ge\hat{c}_{1-\alpha,B}^{m+\ell,n}\,\big|\, Z,\Qhl\Big\}\,\Big|\,\cA_n\bigg\}\\
    &\leq \E\big[\alpha\,\big|\,\cA_n\big]\\
    &=\alpha.
\end{align*}
The first equation builds on tower property; the second equation is due to $\sigma(\cA_n) \subset \sigma\big(Z,\Qhl\big)$. Taking limits as $n\to\infty$ gives the asymptotic statement
\begin{align*}
    \limsup_{n\to\infty}\Pr\Big\{T_n^{m+\ell}\ge \hat c_{1-\alpha,B}^{m+\ell,n}\,\big|\, \cA_n\Big\}\le\alpha.
\end{align*}

\end{proof}

\subsection{Consistency with merged samples}
\label{sec:conditional_consistency}
Similar to \Cref{sec:conditional_validity}, we show that \Cref{thm:partial_boots_consistency} and \Cref{thm:partial_perm_consistency} can be extended to case where we condition on the event $\cA_n=\big\{\phi(\Qcm,\Qhl)=1\big\}$. Since both the partial bootstrap and partial permutation satisfy condition \eqref{eq:unconditional_consistency}, as shown in \Cref{thm:partial_boots_consistency,thm:partial_perm_consistency}, the following result applies to both methods.

\begin{proposition}[Consistency provided merged controls]
\label{prop:consistency_conditional_merged}
 Let $\cA_n = \big\{\phi_f(\Qcm, \Qhl) = 1\big\}$ denote the event of rejecting the fusion null, and assume $\tau \coloneqq\lim_{n\to\infty}\Pr(\cA_n) > 0$. Suppose 
 \begin{equation}
 \label{eq:unconditional_consistency}
\lim_{n\to\infty}\Pr\Big(\phi_\ast\big(\Qcm,\Qhl,\Qtn\big)=1\Big)=1
 \end{equation}under $\cH_1$. Then 
\begin{align*}
\lim_{n\to\infty}\Pr\Big(\phi_\ast\big(\Qcm,\Qhl,\Qtn\big)=1\,|\,\cA_n\Big)=1.
\end{align*}
\end{proposition}
\begin{proof}
There exists $N\in\mathbb{N}_+$ such that $\Pr(\cA_n)>0$ for all $n\ge N$. In what follows, we restrict attention to  $n\ge N$, so that $\Pr( \cdot \;|\; \cA_n)$ is well defined. Using the definition of conditional probability, we have
\begin{align*}
    \Pr\Big(\phi_\ast\big(\Qcm,\Qhl,\Qtn\big)=1\,|\,\cA_n\Big) 
    &=  \frac{\Pr\Big(\big\{\phi_\ast\big(\Qcm,\Qhl,\Qtn\big)=1\big\} \cap\cA_n\Big)}{\Pr(\cA_n)}\\ 
    &= 1 -\frac{\Pr\Big(\big\{\phi_\ast\big(\Qcm,\Qhl,\Qtn\big)=0\big\} \cap\cA_n\Big)}{\Pr(\cA_n)}\\
    &\geq 1 -\frac{\Pr\Big(\phi_\ast\big(\Qcm,\Qhl,\Qtn\big)=0\Big)}{\Pr(\cA_n)}\\
    &\to 1 - \frac{1-\lim_{n\to\infty}\Pr\Big(\phi_\ast\big(\Qcm,\Qhl,\Qtn\big)=1\Big)}{\tau} =1,
\end{align*}
where the last line is resulted in by $\tau>0$ and \eqref{eq:unconditional_consistency}, which is not conditional on $\cA_n$. We thus show the consistency provided the controls are merged. 
\end{proof}

\section{Main proofs}

\subsection{Proof of \Cref{thm:mmd_equiv_test}}
\label{sec:proof_mmd_equiv_test}
\begin{proof}
\Cref{thm:mmd_equiv_test} differs slightly from \citet[Theorem 8]{liu2026kerneltestsequivalence} in that we use the quantile $q_{1-\alpha_f}^{m,\ell}$ of $S_{m,\ell}$ as the critical rule, whereas \citet[Theorem 8]{liu2026kerneltestsequivalence} is stated using a bootstrapped estimate of this quantity. Specifically, their critical value is the corresponding quantile of the conditional of $D_W^2(Q_c^m)$ given $Q_c^m$, where $D_W^2(Q_c^m)$ is defined in \eqref{eq:mmd_equiv_boot_sample}. Nevertheless, a closer inspection of their proof in \citet[Appendix B.2.3]{liu2026kerneltestsequivalence} reveals that they first prove the result using $q_{1-\alpha_f}^{m,\ell}$, and then show that the bootstrapping estimation errors arising from using its bootstrapped estimate instead of $q_{1-\alpha_f}^{m,\ell}$ vanishes as $m \to \infty$ \citep[Lemma 7]{liu2026kerneltestsequivalence}. Consequently, \Cref{thm:mmd_equiv_test} follows by applying the same proof strategy for \citet[Theorem 8]{liu2026kerneltestsequivalence}.
\end{proof}

\subsection{Proof of \Cref{prop:bootstrap_convergence}}
\label{sec:proof_bootstrap_convergence}
\begin{proof}

We define $A \coloneqq \xi_f-\xi_c =  (1-\gamma)\big(\xi_h-\xi_c\big)$, $B \coloneqq \xi_t-\xi_c$. Then $B = 0$, since we assume $Q_c=Q_t$. Furthermore, define $A_n \coloneqq \hat{\xi}_{f,m+\ell} - \hat{\xi}_{c,m}=(1-\gamma)\big(\hat{\xi}_{h,n}-\hat{\xi}_{c,m}\big)$, $B_n \coloneqq \hat{\xi}_{t,n} - \hat{\xi}_{c,m}$. 

With these notations, we have $T=\|A\|^2$, $T_{n,m,\ell}=\|A_n-B_n\|^2$,  $\widetilde{T}_{n,m,\ell}=\|A_n\|^2$. Thus
\begin{align*}
    T_{n,m,\ell}-\widetilde{T}_{n,m,\ell}=\|A_n-B_n\|^2 - \|A_n\|^2 = -2\langle A_n, B_n\rangle+\|B_n\|^2.
\end{align*}
By the central limit theorem (CLT) in Hilbert space (\citealp[Section 9.1]{berlinet2011reproducing}; \citealp[Section 1.2]{wolfer2025variance}):
\begin{align*}
   B_n= B_n - B =\Opr(n^{-1/2})\implies \sqrt{n}\big\|B_n\big\|^2  = \opr(1).
\end{align*}
Similarly, $A_n-A = \Opr(n^{-1/2})$. Thus
\begin{align}
    \Delta=\sqrt{n}\big(T_{n,m,\ell}-\widetilde{T}_{n,m,\ell}\big) &= -2\sqrt{n}\big\langle A_n, B_n\big\rangle+\sqrt{n
    }\big\|B_n\big\|^2\nonumber\\
    &= -2\sqrt{n}\big\langle A_n-A, B_n\big\rangle-2\big\langle A, \sqrt{n}B_n\big\rangle+\opr(1)\nonumber\\
    &=-2\sqrt{n}\big\langle A, B_n\big\rangle+\opr(1) \label{eq:Delta_decomp}.
\end{align}
Note that
$B_n = \frac{1}{n}\sum_{i=1}^n k\big(\cdot, x_i^t\big) - \frac{1}{m} \sum_{j=1}^m k\big(\cdot, x_j^c\big)$, $A = \E_{Y \sim Q_f}k(\cdot, Y) - \E_{Z\sim Q_c}k(\cdot, Z)$. Thus
\begin{align*}
     \big\langle A, B_n \big\rangle =\underbrace{\frac{1}{n}\sum_{i=1}^n\Big[ \E_{Y \sim Q_f}k(x_i^t, Y) - \E_{Z\sim Q_c}k(x_i^t, Z)\Big]}_{\hat{\psi}_t^n} - \underbrace{\frac{1}{m} \sum_{j=1}^m \Big[\E_{Y \sim Q_f}k(x_j^c, Y) - \E_{Z\sim Q_c}k(x_j^c, Z)\Big]}_{\hat{\psi}_c^m}.
\end{align*}
Denote
\begin{align*}
    \psi_t &\coloneqq\E_{X\sim Q_t}\Big[\E_{Y\sim Q_f}k(X, Y) - \E_{Z\sim Q_c}k(X, Z)\Big],\\
    \psi_c &\coloneqq \E_{X\sim Q_c}\Big[\E_{Y\sim Q_f}k(X, Y) - \E_{Z\sim Q_c}k(X, Z)\Big];\\
    \sigma_t^2 &\coloneqq\Var_{X\sim Q_t}\Big[\E_{Y\sim Q_f}k(X, Y) - \E_{Z\sim Q_c}k(X, Z)\Big],\\
    \sigma_c^2 &\coloneqq\Var_{X\sim Q_c}\Big[\E_{Y\sim Q_f}k(X, Y) - \E_{Z\sim Q_c}k(X, Z)\Big].
\end{align*}
We have $\psi_t = \E_{X\sim Q_t}\hat{\psi}_t^n$,  $\psi_c = \E_{X\sim Q_c}\hat{\psi}_c^m$.  Under the assumption $Q_t=Q_c$, we also have $\psi_t = \psi_c$, $\sigma_t^2=\sigma_c^2$. 

By the classical central limit theorem(CLT) and independence of samples,
\begin{align*}
    \Big(\sqrt{n}\big(\hat{\psi}_t^n - \psi_t\big),\, \sqrt{n}\big(\hat{\psi}_c^m - \psi_c\big)\Big)\rightsquigarrow \cN\Big(0,\operatorname{diag}\big(\sigma^2_{c}, \sigma_{c}^2/c_1\big)\Big).
\end{align*}

Note that $\sqrt{n} \big\langle A, B_n \big\rangle= \sqrt{n}\big(\hat{\psi}_t^n - \psi_t\big)-\sqrt{n}\big(\hat{\psi}_c^m - \psi_c\big)$ is a linear continuous map of $\Big(\sqrt{n}\big(\hat{\psi}_t^n - \psi_t\big),\, \sqrt{n}\big(\hat{\psi}_c^m - \psi_c\big)\Big)$. By the continuous mapping theorem, we obtain 
\begin{align}
\label{eq:normal_convergence}
    \sqrt{n} \big\langle A, B_n \big\rangle \rightsquigarrow \cN\big(0,\sigma^2_t + \sigma_c^2/c_1\big).
\end{align}
 
With $\sigma_c^2 = \sigma_t^2$ we get
\begin{align*}
    \Delta \rightsquigarrow \cN\Big(0,4\big(1+1/c_1\big)\sigma_c^2\Big).
\end{align*}

For $\Delta^\ast\coloneqq \sqrt{n}\big(\tnmlstar-\tildetstar\big)$, we follow a similar procedure. Denote $A_n^\ast \coloneqq \hat{\xi}_{f,m+\ell}^\ast - \hat{\xi}_{c,m}^\ast = (1-\gamma)\big(\hat{\xi}_{h,\ell}^\ast - \hat{\xi}_{c,m}^\ast\big)$, $B_n^\ast \coloneqq \hat{\xi}_{t,n}^\ast - \hat{\xi}_{c,m}^\ast$. Thus
\begin{align*}
    \tnmlstar - \tildetstar = \big\|A_n^\ast-B_n^\ast\big\|^2 -\big\|A_n^\ast\big\|^2 = -2\big\langle A_n^\ast , B_n^\ast\big\rangle + \big\|B_n^\ast\big\|^2.
\end{align*}

Conditional on $\Dnml$, we get
\begin{align*}
    B_n^\ast - \E^\ast B_n^\ast =B_n^\ast - \big(\hat{\xi}_{c,m} -\hat{\xi}_{c,m}\big)= B_n^\ast=\operatorname{O}_{{\Pr}^\ast}\big(n^{-1/2}\big).
\end{align*}
This is because both $ \hat{\xi}_{t,n}^\ast$ and $ \hat{\xi}_{c,m}^\ast$ are embeddings of bootstrap samples around  $ \hat{\xi}_{c}^m$. Again we have $\sqrt{n}\|B_n^\ast\|^2 = \operatorname{o}_{{\Pr}^\ast}(1)$. 
Similarly, $A_n^\ast - A_n = \operatorname{O}_{{\Pr}^\ast}(n^{-1/2})$. So
\begin{align*}
    \Delta^\ast = -2\sqrt{n}\big\langle A_n,B_n^\ast\big\rangle + \operatorname{o}_{{\Pr}^\ast}(1).
\end{align*}
Note that
$B_n^\ast = \frac{1}{n}\sum_{i=1}^n k\big(\cdot, x_i^{t,\ast}\big) - \frac{1}{m} \sum_{j=1}^m k\big(\cdot, x_j^{c,\ast}\big)$, $A_n = \frac{1}{m+\ell}\Big[\sum_{i=1}^m k\big(\cdot, x_i^{c}\big)+\sum_{j=1}^\ell k\big(\cdot, x_j^{h}\big)\Big] - \frac{1}{m}\sum_{i=1}^m k\big(\cdot, x_i^{c}\big)$. Thus
\begin{align*}
     \big\langle A_n, B_n^\ast \big\rangle = \hat{\psi}_t^{n,\ast} - \hat{\psi}_c^{m,\ast}
\end{align*}
where
\begin{align*}
    \hat{\psi}_t^{n,\ast} \coloneqq \frac{1}{n}\sum_{i'=1}^n\bigg\{\frac{1}{m+\ell}\Big[\sum_{i=1}^m k(x_{i'}^{t,\ast}, x_i^{c})+\sum_{j=1}^\ell k(x_{i'}^{t,\ast}, x_j^{h})\Big] - \frac{1}{m}\sum_{i=1}^m k(x_{i'}^{t,\ast}, x_i^{c})\bigg\},\\
    \hat{\psi}_c^{m,\ast} \coloneqq \frac{1}{m}\sum_{i'=1}^m\bigg\{\frac{1}{m+\ell}\Big[\sum_{i=1}^m k(x_{i'}^{c,\ast}, x_i^{c})+\sum_{j=1}^\ell k(x_{i'}^{c,\ast}, x_j^{h})\Big] - \frac{1}{m}\sum_{i=1}^m k(x_{i'}^{c,\ast}, x_i^{c})\bigg\}.
\end{align*}
Denote
\begin{align*}
    \hat\psi\coloneqq \E^\ast_{X\sim \Qcm}\bigg\{ \frac{1}{m+\ell}\Big[\sum_{i=1}^m k(\cdot, x_i^{c})+\sum_{j=1}^\ell k(\cdot, x_j^{h})\Big] - \frac{1}{m}\sum_{i=1}^m k(\cdot, x_i^{c})\bigg\},
\end{align*}
we apply bootstrap CLT conditional on $\Dnml$:
\begin{align*}
    \Big(\sqrt{n}(\hat{\psi}_t^{n,\ast}-\hat{\psi}), \sqrt{n}(\hat{\psi}_c^{m,\ast}-\hat{\psi})\Big) \rightsquigarrow_{\Pr^\ast} \mathcal{N}\Big(0, \operatorname{diag}\big({\sigma_c^\ast}^2,{\sigma_c^\ast}^2/c_1\big)\Big),
\end{align*}
where ${\sigma_c^\ast}^2 \coloneqq \Var^\ast_{X\sim \Qcm}\bigg\{ \frac{1}{m+\ell}\Big[\sum_{i=1}^m k\big(X, x_i^{c}\big)+\sum_{j=1}^\ell k\big(X, x_j^{h}\big)\Big] - \frac{1}{m}\sum_{i=1}^m k\big(X, x_i^{c}\big)\bigg\}$, when the number of bootstrap samples $B\to\infty$.
With continuous mapping theorem we obtain
\begin{align*}
    \Delta^\ast \rightsquigarrow_{\Pr^\ast} \cN\Big(0,4\big(1+1/c_1\big){\sigma_c^\ast}^2\Big). 
\end{align*}
 
The convergence of $\Delta$ and $\Delta^\ast$, as well as ${\sigma_c^\ast}^2\rightarrow \sigma_c^2$ together give us
\begin{align*}
    \sup_{z\in\R} \bigg|{\Pr}^\ast\Big\{\sqrt{n}\big(\tnmlstar-\tildetstar\big)\leq z\Big\} - \Pr\Big\{\sqrt{n}\big(\tnml-\tildet\big)\leq z\Big\} \bigg| 
    \overset{a.s.}{\longrightarrow}0.
\end{align*}

\end{proof}

\subsection{Proof of \Cref{thm:partial_boots_validity}}
\label{sec:proof_partial_boots_validity}
\begin{proof}
We write $F_n^\ast(z) = {\Pr}^\ast\Big\{\sqrt{n}\big(\tnmlstar-\tildetstar\big)\leq z\Big\}$ as the conditional bootstrap CDF, and $F_n(z) = \Pr\Big\{\sqrt{n}\big(\tnml-\tildet\big)\leq z\Big\}$ as the unconditional CDF. Also write $F(\cdot)$ as the CDF of $\mathcal{N}(0, 4(1+1/c_1)\sigma_c^2)$. By \Cref{prop:bootstrap_convergence} and Portmanteau theorem, we have 
\begin{equation*}
F_n(z)\overset{a.s.}{\longrightarrow}F(z)
\end{equation*}
at every continuity point $z$ of $F$. Moreover, we have the uniform consistency
\begin{equation}
    \sup_{z\in \R} \Big|F_n^\ast(z)-F_n(z)\Big| \overset{a.s.}{\longrightarrow}0.
    \label{eq:cdf_uniform_consistency}
\end{equation}
Define $q_{n,1-\alpha}^\ast \coloneqq \inf_{z\in\R}\{z:F_n^\ast(z)\geq1-\alpha\}$, $q_{n,1-\alpha} \coloneqq \inf_{z\in\R}\{z:F_n(z)\geq1-\alpha\}$ and $q_{1-\alpha} \coloneqq \inf_{z\in\R}\{z:F(z)\geq1-\alpha\}$.

Decompose
\begin{align}
    \Pr(\Delta > q_{1-\alpha,B}^{m+\ell,n})
    &\leq 
    \underbrace{\Pr\big\{\Delta > q_{1-\alpha,B}^{m+\ell,n}\big\} - \Pr\nolimits^\ast\big\{\Delta^\ast > q_{1-\alpha,B}^{m+\ell,n}\big\} }_{\epsilon_{n,B}}
    + \Pr\nolimits^\ast\big\{\Delta^\ast > q_{1-\alpha,B}^{m+\ell,n}\big\}.
    \label{eq:empirical_cdf_decompose}
\end{align}

We have for any fixed $B$:
\begin{equation}
    \epsilon_{n,B}\leq \epsilon_n\coloneqq  \sup_{z\in \R} \Big|F_n^\ast(z)-F_n(z)\Big|,
    \label{eq:error_bound}
\end{equation}
while $\epsilon_n \overset{a.s.}{\longrightarrow}0$ as $n\to\infty$ by \eqref{eq:cdf_uniform_consistency}.

Recall that $F_n^\ast$ is the population CDF of the bootstrap draws $\Delta^\ast$; correspondingly, we denote $F_{n,B}$ as the empirical CDF of $\Delta^\ast$. Define $\delta_B' = \sqrt{\frac{1}{2B}\log(2B)}$. By the Dvoretzky–Kiefer–Wolfowitz inequality \citep{dvoretzky1956asymptotic}, the event $\cA = \Big\{ \sup_u \big| F_{n,B}(u) - F_n^\ast(u) \big| \leq \delta_B'\Big\}$ occurs with probability $\Pr^\ast(\cA) \geq 1 - B^{-1}$. Moreover, on $\cA$,
\begin{align*}
    q_{1-\alpha,B}^{m+\ell,n}
    = 
    \inf\big\{ u \in \R: 1-\alpha \leq F_{n,B}(u) \big\}
    &\geq\big\{ u \in \R: 1-\alpha \leq F_n^\ast(u) + \delta_B' \big\}
    \\
    &=\inf\big\{ u \in \R: 1-\alpha-\delta_B' \leq F_n^\ast(u) \big\}
    \\
    &=q_{n,1-\alpha-\delta_B'}.
\end{align*}
Thus 
\begin{equation}
    \Pr\nolimits^\ast\big\{\Delta^\ast > q_{1-\alpha,B}^{m+\ell,n} \big\}
    \leq 
    \Pr\nolimits^\ast\Big(\big\{\Delta^\ast > q_{n,1-\alpha-\delta_B'} \big\}\,,\, \cA\Big) + \Pr\nolimits^\ast(\cA)
    \leq \alpha + \delta_B' + \frac{1}{B}
    \eqqcolon \alpha + \delta_B,
    \label{eq:bootstrap_cdf_bound}
\end{equation}
where $\delta_B\to 0$ as $B\to \infty$.

Combining \eqref{eq:empirical_cdf_decompose}, \eqref{eq:error_bound} and \eqref{eq:bootstrap_cdf_bound}, we have 
\begin{align*}
\limsup_{n\to\infty}\lim_{B\to\infty}\Pr\big(\Delta > q_{1-\alpha,B}^{m+\ell,n} \big)
    \leq \alpha,
\end{align*}
which completes the proof.
\end{proof}

\subsection{Proof of \Cref{thm:partial_boots_consistency}}
\label{sec:proof_partial_boots_consistency}
\begin{proof}
    We start by using $q_{1-\alpha}^\ast$, which is the $(1-\alpha)$-quantile of the distribution of $\sqrt{n}\big(\tnmlstar-\tildetstar\big)$. We write $ T \coloneqq D^2\big(Q_f,Q_t\big) =\xi_f-\xi_t$. Using the notations in \Cref{sec:proof_bootstrap_convergence}, we denote $B \coloneqq \xi_t-\xi_c$, $A \coloneqq \xi_f - \xi_c$. Define the population gap
    \begin{align*}
        \Delta_\infty \coloneqq T-\widetilde{T} = \big\|\xi_f - \xi_t\big\|^2-\big\|\xi_f-\xi_c\big\|^2 = \|A-B\|^2-\|A\|^2 = -2\langle A,B \rangle + \|B\|^2.
    \end{align*}
    
Assume the alternative hypothesis and that $2(1-\gamma)D\big(Q_h,Q_c\big)\cos\beta < \,D(Q_c, Q_t)$. Since $A = \xi_f-\xi_c = (1-\gamma)(\xi_h-\xi_c)$,  we have $2(1-\gamma)\|A\|\cos\beta <\|B\|$. Consequently,
\begin{align*}
    \Delta_\infty =\|B\|^2-2\langle A,B\rangle = \|B\|^2-2\|A\|\|B\|\cos\beta>0.
\end{align*}

With the assumption that all kernel mean embeddings have finite first RKHS moment, the Bochner integral is well defined and the empirical embedding satisfies a strong law of large numbers (\citealp[Theorem 1.1]{cltbanach}; \citealp[Corollary 7.10]{ledoux2013probability}). Together with \Cref{lem:merged_sample_kme}, we have $\tnml \rightarrow T$, $\tildet\rightarrow \widetilde{T}$ almost surely. Hence
\begin{align*}
    \big(\tnml-\tildet\big) \to \Delta_{\infty} >0.
\end{align*}
Thus, 
\begin{align*}
    \Delta = \sqrt{n}\big(\tnml-\tildet\big) \to \infty.
\end{align*}
As proved in \Cref{sec:proof_bootstrap_convergence}, the distribution of $\sqrt{n}(\tnmlstar-\tildetstar)$ converges to $\cN\big(0,4(1+1/c_1)\sigma_c^2\big)$. Thus $q_{1-\alpha}^\ast = \Opr(1)$. By the conditional Glivenko-Cantelli, we have the quantile of bootstrap samples $\hat{q}_{1-\alpha,B}^{m+\ell,n} \rightarrow q_{1-\alpha}^\ast$. We thus obtain
\begin{align*}
    \lim_{n \to \infty} {\Pr}\big(\Delta>\hat{q}_{1-\alpha,B}^{m+\ell,n}\,\big|\,\cH_1\big) =1.
\end{align*}

Note that $\Delta_\infty>0$ is the sufficient and necessary condition for consistency to hold. We next show that if $\Delta_\infty = 0$, the test is not consistent.

Using the same notations in \Cref{sec:proof_partial_boots_validity}, we have
\begin{align*}
    A_n-A = \Opr(n^{-1/2}), \quad B_n-B=\Opr(n^{-1/2}),
\end{align*}
and 
\begin{align*}
\widetilde{V}_n &\coloneqq \Delta - \Delta_\infty\\
&= \sqrt{n}\Big(-2\big\langle A_n,B_n\big\rangle+\big\|B_n\big\|^2+2\big\langle A,B\big\rangle-\|B\|^2\Big)\\
&= -2\sqrt{n}\big\langle A_n-A,B\big\rangle
-2\sqrt{n}\big\langle A,B_n-B\big\rangle
+2\sqrt{n}\big\langle B,B_n-B\big\rangle
+\opr(1)\\
&= 2\sqrt{n}\Big(
\underbrace{\big\langle B-A,\hat{\xi}_{t,n}-\xi_t\big\rangle}_{\widetilde{\psi}_t^n-\widetilde{\psi}_t}-\underbrace{\big\langle B,\hat{\xi}_{f,m+\ell}-\xi_f\big\rangle}_{\widetilde{\psi}_f^{m+\ell}-\widetilde{\psi}_f}
+\underbrace{\big\langle A,\hat{\xi}_{c,m}-\xi_c\big\rangle}_{\widetilde{\psi}_c^m-\widetilde{\psi}_c}\Big)+\opr(1),
\end{align*}
where we have 
\begin{align*}
    \widetilde{\psi}_t(X) \coloneqq \E_{Y\sim Q_t}k(X,Y) - \E_{Z\sim Q_f}k(X,Z),\\
    \widetilde{\psi}_f(X) \coloneqq \E_{Y\sim Q_t}k(X,Y) - \E_{Z\sim Q_c}k(X,Z),\\
    \widetilde{\psi}_c(X) \coloneqq \E_{Y\sim Q_f}k(X,Y) - \E_{Z\sim Q_c}k(X,Z).
\end{align*}
% With the notations above, we also have 
% \begin{align*}
%     \widetilde{\psi}_t = \E_{X\sim Q_t} \widetilde{\psi}_t(X), \quad \widetilde{\psi}_f = \E_{X\sim Q_f} \widetilde{\psi}_f(X), \quad   \widetilde{\psi}_c = \E_{X\sim Q_c} \widetilde{\psi}_c(X).
% \end{align*}
Denote the corresponding variances 
\begin{align*}
    \widetilde{\sigma}_t^2 \coloneqq \Var_{X\sim Q_t} \widetilde{\psi}_t(X), \quad \widetilde{\sigma}_f^2 \coloneqq \Var_{X\sim Q_f} \widetilde{\psi}_f(X), \quad \widetilde{\sigma}_c^2 = \Var_{X\sim Q_c} \widetilde{\psi}_c(X).
\end{align*}
Let the empirical counterparts be
\begin{align*}
    \widetilde{\psi}_t^n\coloneqq\frac{1}{n}\sum_{i=1}^n\widetilde{\psi}_t(x_i^t), \quad \widetilde{\psi}_f^{m+\ell}\coloneqq\frac{1}{m+\ell}\Big[\sum_{i=1}^{m}\widetilde{\psi}_c(x_i^c)+ \sum_{i=1}^{\ell}\widetilde{\psi}_h(x_i^h)\Big], \quad \widetilde{\psi}_c^m\coloneqq\frac{1}{m}\sum_{i=1}^m\widetilde{\psi}_c(x_i^c).    
\end{align*}
By the CLT in Hilbert space, the linearity and continuous mapping, we get
\begin{align*}
    \widetilde{V}_n \rightsquigarrow_{\Pr} \cN\Big(0, 4\big(\widetilde{\sigma}_t^2+\widetilde{\sigma}_c^2/c_1 + \widetilde{\sigma}_f^2/(c_1+c_2)\big)\Big).
\end{align*}

Thus when $\Delta_\infty=0$, we have $\Delta \rightsquigarrow_{\Pr} \cN(0,\sigma_1^2)$ where $\sigma_1^2\coloneqq 4\big(\widetilde{\sigma}_t^2+\widetilde{\sigma}_c^2/c_1 + \widetilde{\sigma}_f^2/(c_1+c_2)\big)$ can be different from $\sigma_0^2\coloneqq4(1+1/c_1)\sigma_c^2$. As a result, we have
\begin{align*}
    \Pr\big\{\Delta >q_{1-\alpha}^\ast\big\} =1-\Phi\Big(\frac{\sigma_0}{\sigma_1}z_{1-\alpha}\Big)
\end{align*}
where $\Phi(\cdot)$ is the standard Gaussian CDF and $z_{1-\alpha}$ is the $(1-\alpha)$-quantile of standard Gaussian. Since $\frac{\sigma_0}{\sigma_1}z_{1-\alpha}>0$, $1-\Phi\big(\frac{\sigma_0}{\sigma_1}z_{1-\alpha}\big) < 0.5$. By replacing $q_{1-\alpha}^\ast$ with $\hat{q}_{1-\alpha,B}^{m+\ell,n}$ and taking $n\to\infty$, we obtain $ \lim_{n \to \infty} {\Pr}\big(\Delta>\hat{q}_{1-\alpha,B}^{m+\ell,n}\,\big|\,\cH_1\big) < 0.5$, thus the test is not consistent.
\end{proof}

\subsection{Proof of \Cref{thm:partial_perm_validity}}
\label{sec:proof_partial_perm_validity}
\begin{proof}

 Write
$Z:=\{Z_1,\dots,Z_{m+n}\}$ for the pooled sample $Q_c^m \cup Q_t^n$. Let $G$ be the set of all permutations $\pi$ of $I_Z\coloneqq\{1,\ldots,m+n\}$. Denote 
the observed control indices as $I_0$. 

For any $\pi \in G$, write the permuted control index set $I_c(\pi)$, permuted treatment index set $I_t(\pi)$, with $\big|I_c(\pi)\big|=m$, $\big|I_t(\pi)\big|=n$. Define the permuted empirical measures  
\begin{align*}
    Q_{c,\pi}^m\coloneqq\frac{1}{m}\sum_{i\in I_c(\pi)}\delta_{Z_i}, \quad  Q_{t,\pi}^n\coloneqq\frac{1}{n}\sum_{i\in I_t(\pi)}\delta_{Z_i}, 
\end{align*}
with the fused control 
\begin{align*}
    Q_{f,\pi}^{m+\ell}:=\frac{m}{m+\ell}Q_{c,\pi}^m+\frac{\ell}{m+\ell}Q_{h}^\ell,
\end{align*}
and let $T_{n,m,\ell}^\pi:=D^2(Q_{f,\pi}^{m+\ell}, Q_{t,\pi}^n)$. The observed statistic is $\tnml= T_{n,m,\ell}^{\pi^0}$ where $\pi^0$ is the identity $\pi^0(i)=i$.

Under $H_0:Q_c=Q_t$, the pooled sample $Z_1,\ldots,Z_{m+n}\iid Q_c$ has a joint distribution that is invariant under permutation. Moreover, since $\Qhl$ is independent of $Z$ and not permuted, the joint distribution, on the population level, satisfies
\begin{align*}
    \big(Z,\Qhl\big) \stackrel{d}{=}\big(\pi\cdot Z,\Qhl\big), \quad \forall \pi\in G.
\end{align*}
This is precisely the randomization hypothesis: the joint law of the data is invariant under the group $G$. Since $\tnml^\pi$ is a measurable function of $(\pi \cdot Z, \Qhl)$, the collection $\big\{\tnml^\pi: \pi \in G\big\}$ is invariant under the group so satisfies the assumptions of the randomized test theorem \citep[Definition 17.2.1; Theorem 17.2.2]{lehmann2022testing}.

% Let the $(1-\alpha)$–quantile of group $G$ be
% \begin{align*}
%     c_{1-\alpha}^{m+\ell,n} \coloneqq \inf\Big\{u:\ \frac{1}{|G|}\sum_{\pi\in G}\mathbbm{1}\big\{\tnml^\pi\le u\big\}\ \ge\ 1-\alpha\Big\}.
% \end{align*}

% By the classical randomization result \citep[Theorem 17.2]{lehmann2022testing}, any test that compares the observed statistic $\tnml$ to the group $(1-\alpha)-$quantile of $\big\{\tnml^\pi: \pi \in G\big\}$ under an invariant distribution has size at most $\alpha$. Therefore, under $\cH_0$, 
% \begin{align*}
%     \Pr\big\{\tnml \geq c_{1-\alpha}^{m+\ell,n} \leq \alpha\big\},
% \end{align*}
% for every $n$, $m$, $\ell$.

In our algorithm, we only samples a subset of permutations $\pi^1, \ldots, \pi^B$ of size $B$ uniformly from $G$, and calculate $\hat{c}_{1-\alpha,B}^{m+\ell,n}$ as the empirical quantile of $\big\{\tnml^{\pi_b}\big\}_{b=0}^B$. Conditional on the data $\Dnml$, the collection $\{\tnml^{\pi_b}\}_{b=1}^B$ is exchangeable. By exchangeability, the rank of $\tnml=\tnml^{\pi_0}$ among $\big\{\tnml^{\pi_b}\big\}_{b=1}^B$ is uniform on $\{1,\ldots, B+1\}$. We choose $\hat{c}_{1-\alpha,B}^{m+\ell,n}$ so that at most an $\alpha-$fraction of the sample with value larger than or equal to it. Hence, by the Monte Carlo randomization tests \citep{lehmann2022testing}, we have
\begin{align}
    \Pr\Big\{\tnml \geq \hat{c}_{1-\alpha,B}^{m+\ell,n} \,\big|\, Z, \Qhl\Big\}\leq\alpha.
\label{eq:partial_perm_conditional_validity}
\end{align}
Finally, take expectations over $\big(Z,\Qhl\big)$ we get the unconditional result under $\cH_0$:
\begin{align*}
        \Pr\big\{\tnml \geq \hat{c}_{1-\alpha,B}^{m+\ell,n} \big\}\leq\alpha,
\end{align*}
for all $n$, $m$, $\ell$, $B$.
\end{proof}

\subsection{Proof of \Cref{thm:partial_perm_consistency}}
\label{sec:proof_partial_perm_consistency}
We first introduce the Hadamard differentiability of the squared MMD.
\begin{lemma}[Hadamard differentiability of the squared MMD]
\label{lem:hadamard_mmd}
For $\gamma \in (0,1)$, the squared MMD function 
$D^2\big((1-\gamma)Q_c+\gamma Q_h, Q_t\big)$ is Hadamard differentiable with respect to probability measures $Q_c$, $Q_h$ and $Q_t$.
\end{lemma}
\begin{proof}[Proof of \Cref{lem:hadamard_mmd}]
    Define $\Psi:\cH_k^3\times[0,1]\rightarrow \R$,
\begin{align*}
    \Psi(\xi_c,\xi_h,\xi_t;\gamma) \coloneqq \Big\|\xi_t- \big(\gamma\xi_c+(1-\gamma)\xi_h\big)\Big\|_{\cH_k}.
\end{align*}
  Since $\Psi$ is a quadratic form on a Hilbert product space, it is Fréchet, hence Hadamard, differentiable with respect to $(\xi_c,\xi_h,\xi_t)$. Composing with the linear embedding map $R \mapsto \xi_R$,  where $R$ is a probability measure, gives Hadamard differentiability of $(Q_c,Q_t,Q_h) \mapsto T(Q_c,Q_t,Q_h)= \Psi(\xi_c,\xi_h,\xi_t;\gamma)$.
\end{proof}
Now we provide the proof for  \Cref{thm:partial_perm_consistency}.

\begin{proof}[Proof of \Cref{thm:partial_perm_consistency}]
 Denote $Q_d \coloneqq\lambda Q_c+ (1-\lambda)Q_t$, $Q_{f,d} = \gamma Q_d+(1-\gamma)Q_h$. The corresponding population kernel mean embeddings are written as $\xi_d \coloneqq (1-\lambda) \xi_c + \lambda\xi_t$,  $\xi_{f,d} \coloneqq \gamma \xi_d + (1-\gamma)\xi_h$.

Denote the pooled empirical measure of $Z\coloneqq(Z_1,\ldots,Z_{m+n})$ (defined in \Cref{sec:proof_partial_perm_validity}) as $Q_{d}^{m+n} \coloneqq (1-\lambda) \Qcm+ \lambda \Qtn$. In each permutation $b$, let $Q_{c,\pi^b}^m$ be the empirical measure of the $m$ observations drawn without replacement from the $m+n$ pooled observations, and similarly $Q_{t,\pi^b}^n$ for the remaining $n$. The fused control in permutation $b$ is thus $Q_{f,\pi^b}^{m+\ell}=\gamma Q_{c,\pi^b}^m + (1-\gamma)Q_{h,b}^{\ell}$; the empirical measure of the $b$-th pooled data is denoted $Q_{d,\pi^b}^{m+n} \coloneqq (1-\lambda)Q_{c,\pi^b}^m + \lambda Q_{t,\pi^b}^n$. The statistic $T_{n,m,\ell}^{b}$ of the $b$-th permutation is thus $D^2(Q_{f,\pi^b}^{m+\ell}, Q_{t,\pi^b}^{n})$.

Conditional on the pooled sample, by the extended Glivenko--Cantelli theorem for finite populations \citep[Corollary 3.2]{motoyama2024extended}, the permutation empirical measures satisfy
\begin{align*}
Q_{c,\pi^b}^m \Rightarrow Q_d^{m+n}, 
\qquad 
Q_{t,\pi^b}^n \Rightarrow Q_d^{m+n},
\end{align*}
almost surely as $m,n\to\infty$, where $Q_d^{m+n}$ denotes the empirical distribution of the pooled data.

For $Q_h$, the empirical kernel mean embedding satisfies $\xi_h^\ell \to \xi_h$ almost surely. By linearity of kernel mean embeddings and the continuous mapping theorem, it follows that
\begin{align*}
    Q_{f,\pi^b}^{m+\ell} \Rightarrow Q_{f,d}.
\end{align*}

As $D^2(\cdot,\cdot)$ is Hadamard differentiable according to \Cref{lem:hadamard_mmd}, together with continuous mapping theorem we get
\begin{align*}
T_{n,m,\ell}^{b}=D^2\big(Q_{f,\pi^b}^{m+\ell}, Q_{t,\pi^b}^{n}\big) \to D^2\big(Q_{f,d}, Q_d\big).
\end{align*}

So the partial permutation yields $T_{n,m,\ell}^b$ that converges to 
\begin{align*}
    D^2\big(Q_{f,d}, Q_d\big)&=\Big\|\frac{m}{m+\ell}\xi_d+\frac{\ell}{m+\ell}\xi_h-\xi_d\Big\|^2\\
    &=\big\|(1-\gamma)(\xi_h-\xi_d)\big\|^2\\
    & =\big\|(1-\gamma)(\xi_h-\xi_c) - \lambda(1-\gamma) (\xi_t-\xi_c)\big\|^2.
\end{align*}
While $\tnml$ converges to 
\begin{align*}
    \big\|(1-\gamma)(\xi_h-\xi_c) - (\xi_t-\xi_c)\big\|^2.
\end{align*}
We want them to be separable, which means
\begin{align*}
    \underbrace{\big\|(1-\gamma)(\xi_h-\xi_c) - \lambda(1-\gamma) (\xi_t-\xi_c)\big\|^2}_{(1)}<\underbrace{\big\|(1-\gamma)(\xi_h-\xi_c) - (\xi_t-\xi_c)\big\|^2}_{(2)}.
\end{align*}
 We have
\begin{align*}
  (1)-(2) &=  \big(1-\lambda(1-\gamma)\big)\Big\{2(1-\gamma)\big\langle\xi_h-\xi_c,\xi_t-\xi_c\big\rangle-\big(1+\lambda(1-\gamma)\big)\big \|\xi_t-\xi_c\|^2\Big\}.
\end{align*}
Since $\lambda \in (0,1)$, $\gamma \in (0,1)$, we have $1-\lambda(1-\gamma)>0$. Thus   
\begin{align*}
(1)<(2) \iff 2(1-\gamma)\big\langle\xi_h-\xi_c,\xi_t-\xi_c\big\rangle-\big(1+\lambda(1-\gamma)\big)\big \|\xi_t-\xi_c\|^2 <0.
\end{align*}
Since we have $\big\|\xi_t-\xi_c\big\|\neq 0$, the above inequality stands if and only if $\xi_h-\xi_c = 0$ or $2\cos\beta\, \big\|\xi_h-\xi_c\big\|< \Big(\frac{1}{1-\gamma}+\lambda \Big) \big\|\xi_t-\xi_c\big\|$.
\end{proof}

\subsection{Proof of \Cref{thm:validity_general}}
\label{pf:thm:validity_general}
\begin{proof}
    Assuming first that $\cH_0^f: D\big(Q_h,Q_c\big) \geq \theta$ holds, we can rewrite the Type-I error rate using law of total probability as
    \begin{align*}
        &\,
        \Pr\Big(\phi_\ast\big(\Qcm, \Qhl, \Qtn\big) = 1 \,\big|\, \cH_0,\, \cH_0^f\Big) 
        \\
        &=
        \underbrace{\Pr\Big(\phi_\ast\big(\Qcm, \Qhl, \Qtn) = 1 \,\big|\, \cH_0,\, \cH_0^f,\, \phi_f\big(\Qcm, \Qhl\big) = 1\Big)}_{T_{11}} 
        \underbrace{\Pr\Big(\phi_f\big(\Qcm, \Qhl\big) = 1 \,\big|\, \cH_0,\, \cH_0^f\Big)}_{T_{12}}
        \\
        &\quad
        + \underbrace{\Pr\Big(\phi_\ast\big(\Qcm, \Qhl, \Qtn\big) = 1 \,\big|\, \cH_0,\, \cH_0^f,\, \phi_f\big(\Qcm, \Qhl\big) = 0\Big)}_{T_{13}} 
        \underbrace{\Pr\Big(\phi_f\big(\Qcm, \Qhl\big) = 0 \,\big|\, \cH_0,\, \cH_0^f\Big)}_{T_{14}}
        \\
        \,&=\,
        T_{11} T_{12} + T_{13} \big(1 - T_{12}\big)
        ,
    \end{align*}
    where the last line holds since $T_{14} = 1 - T_{12}$. Intuitively, $T_{11}$ is the probability of rejecting the causality test given that the fusion test is also rejected, i.e., $\Qcm$ and $\Qhl$ are merged, and $T_{12}$ is the probability of rejecting the fusion test. The terms $T_{13}$ and $T_{14}$ have similar interpretations. We will bound each term separately. We first observe that the limiting value of $T_{12}$ is
    \begin{align}
        T_{12}^\infty \,\coloneqq\, 
        \lim_{n \to \infty} T_{12} 
        \,\stackrel{(a)}{=}\, \lim_{n \to \infty} \Pr\Big(\phi_f\big(\Qcm, \Qhl\big) = 1 \,\big|\, \cH_0^f \Big) 
        \,\stackrel{(b)}{=}\,
        \begin{cases}
            0 \,,   & \quad D\big(Q_c, Q_h\big) > \theta \,, \\
            a, & \quad D\big(Q_c, Q_h\big) = \theta ,
        \end{cases}
        \label{eq:T12_limit}
    \end{align}
    where $a$ is defined in \eqref{eq:fusion_test_validity}, step $(a)$ holds since the events $\big\{\phi_f(\Qcm, \Qhl) = 1\big\}$ and $\big\{\cH_0: Q_c = Q_t\big\}$ are independent, and $(b)$ holds by noting that the LHS is the limiting Type-I error rate of the fusion test under $\cH_0^f: D\big(Q_h,Q_c\big) \geq \theta$, which is either $0$ or $\alpha$ by the assumption \eqref{eq:fusion_test_validity}. Moreover, we can bound $T_{13}$ as
    \begin{align*}
        \limsup_{n \to \infty} T_{13}
        &=
        \limsup_{n \to \infty} \Pr\Big(\phi_\ast(\Qcm, \Qhl, \Qtn) = 1 \,\big|\, \cH_0\Big)
        \leq
        \alpha
        ,
    \end{align*}
    where the equality holds since $\phi_f\big(\Qcm, \Qhl\big) = 0$ means no merging, so $\phi_\ast\big(\Qcm, \Qhl, \Qtn\big)$ does not depend on $\Qhl$, and the last step follows by the assumed Type-I error control of the causality test $\phi_\ast$, namely \eqref{eq:causal_test_validity_no_merge}. We now bound $T_{11}$, the Type-I error of the causality test given merging, by treating it separately for $D\big(Q_h,Q_c\big) > \theta$ and $D\big(Q_h,Q_c\big) = \theta$. When $D\big(Q_h,Q_c\big) > \theta$, we bound the probability $T_{11}$ by 1 and combine with the above results to yield
    \begin{align}
        \limsup_{n \to \infty} \Pr\Big(\phi_\ast\big(\Qfml, \Qtn\big) = 1 \,\big|\, \cH_0,\, D\big(Q_c, Q_h\big) > \theta\Big) 
        &\leq
        1 \times \lim_{n \to \infty} T_{12} + \limsup_{n \to \infty} T_{13} \times \lim_{n \to \infty} (1 - T_{12}) \nonumber \\
        &\leq
        0 + \alpha \times 1 \nonumber \\
        &=
        \alpha.
        \label{eq:bound_case_1}
    \end{align}
    where in the first step we have used the inequality $\limsup_{n \to \infty}f_ng_n \leq \big(\limsup_{n \to \infty}f_n\big)\big(\limsup_{n \to \infty}g_n\big)$ for non-negative sequences $f_n,g_n$, and the second line holds due to \eqref{eq:T12_limit} and that we have assumed $D(Q_c, Q_h) > \theta$.
    
    When instead $D\big(Q_c, Q_h\big) = \theta$, we invoke assumption \eqref{eq:causal_test_validity} to give
    \begin{align*}
        \limsup_{n \to \infty} T_{11}
        =
        \limsup_{n \to \infty} \Pr\Big(\phi_\ast\big(\Qcm, \Qhl, \Qtn\big) = 1 \,\big|\, \cH_0, \, D\big(Q_c, Q_h\big)=\theta, \, \phi_f\big(\Qcm, \Qhl\big)=1\Big)
        \leq
        \alpha.
    \end{align*}
    In this case, the overall Type-I error rate can be bounded as
    \begin{align*}
        &
        \limsup_{n \to \infty} \Pr\Big(\phi_\ast(\Qfml, \Qtn) = 1 \,\big|\, \cH_0,\, D\big(Q_h,Q_c\big) = \theta\Big) 
        \\
        &\leq
        \limsup_{n \to \infty} T_{11} \times \lim_{n \to \infty} T_{12}
        + \limsup_{n \to \infty} T_{13} \times \lim_{n \to \infty} \big(1 - T_{12}\big)
        \\
        &\leq
        \alpha T_{12}^\infty + \alpha \big(1 - T_{12}^\infty\big)
        \\
        &=
        \alpha
        .
        \tagaligneq
        \label{eq:bound_case_2}
    \end{align*}
    This proves the desired result when $\cH_0^f: D\big(Q_h,Q_c\big) \geq \theta$.

    We now prove the claimed result assuming $\cH_1^f: D\big(Q_h,Q_c\big) < \theta$. Decomposing the Type-I error probability in a similar way as before, we have
    \begin{align*}
        &
        \Pr\Big(\phi_\ast\big(\Qcm, \Qhl, \Qtn\big) = 1 \,\big|\, \cH_0,\, \cH_1^f\Big) 
        \\
        \,&=\,
        \underbrace{\Pr\Big(\phi_\ast\big(\Qcm, \Qhl, \Qtn\big) = 1 \,\big|\, \cH_0,\, \cH_1^f,\, \phi_f\big(\Qcm, \Qhl\big) = 1\Big)}_{T_{21}}
        \underbrace{\Pr\Big(\phi_f\big(\Qcm, \Qhl\big) = 1 \,\big|\, \cH_0,\, \cH_1^f\Big)}_{T_{22}}
        \\
        &\,\quad
        + \underbrace{\Pr\Big(\phi_\ast\big(\Qcm, \Qhl, \Qtn\big) = 1 \,\big|\, \cH_0,\, \cH_1^f,\, \phi_f\big(\Qcm, \Qhl\big) = 0\Big)}_{T_{23}}
        \underbrace{\Pr\Big(\phi_f\big(\Qcm, \Qhl\big) = 0 \,\big|\, \cH_0,\, \cH_1^f\Big)}_{T_{24}}
        \\
        &=
        T_{21} + (T_{23} - T_{21}) T_{24}
        ,
    \end{align*}
    where in the last line we have used $T_{22} = 1 - T_{24}$. Under the assumed condition \eqref{eq:fusion_test_validity} on the fusion test, we have $\lim_{n \to \infty} T_{24} = 0$. Since furthermore $\big|T_{23} - T_{21}\big|$ is bounded, this implies that $\lim_{n \to \infty} \big(T_{23} - T_{21}\big) T_{24} = 0$. Moreover, under the event $\big\{\phi_f\big(\Qcm, \Qhl\big) = 1\big\}$, the empirical sample $\Qfml$ is the merged sample based on the independent samples $\Qcm$ and $\Qhl$. This implies that 
    \begin{align*}
        \limsup_{n \to \infty}T_{21}
        =
        \limsup_{n \to \infty} \Pr\Big(\phi_\ast\big(\Qcm, \Qhl, \Qtn\big) = 1 \,\big|\, \cH_0, \, D\big(Q_c, Q_h\big) < \theta \Big) 
        \leq 
        \alpha,
    \end{align*}
    where the inequality holds by the validity of the causality test assumed in \eqref{eq:causal_test_validity}. Combining these results gives
    \begin{align*}
        \limsup_{n \to \infty} \Pr\Big(\phi_\ast\big(\Qcm, \Qhl, \Qtn\big) = 1 \,\big|\, \cH_0,\, \cH_1^f\Big) 
        \leq
        \limsup_{n \to \infty}T_{21} + \lim_{n \to \infty}\big(T_{23} - T_{21}\big)\,T_{24}
        \leq
        \alpha.
    \end{align*}
    The above inequality, together with \eqref{eq:bound_case_1} and \eqref{eq:bound_case_2}, completes the proof.
\end{proof}

\subsection{Proof of \Cref{thm:consistency_general}}
\label{sec:proof_thm_consistency_general}
\begin{proof}
Fix $Q_c, Q_h, Q_t$ with $Q_c \neq Q_t$. Decomposing the probability of rejection in the same way as in \Cref{pf:thm:validity_general}, we have
\begin{align*}
    \Pr\Big(\phi_\ast\big(\Qcm, \Qhl, \Qtn\big) = 1\Big) 
    \,&=\,
    \underbrace{\Pr\Big(\phi_\ast\big(\Qcm, \Qhl, \Qtn\big) = 1 \,\big|\, \phi_f\big(\Qcm, \Qhl\big) = 1\Big)}_{T_{1}} 
    \underbrace{\Pr\Big(\phi_f\big(\Qcm, \Qhl\big) = 1\Big)}_{T_{2}}
    \\
    &\,\quad
    + \underbrace{\Pr\Big(\phi_\ast\big(\Qcm, \Qhl, \Qtn\big) = 1 \,\big|\, \phi_f\big(\Qcm, \Qhl\big) = 0\Big)}_{T_{3}} 
    \underbrace{\Pr\Big(\phi_f(\Qcm, \Qhl) = 0\Big)}_{1-T_2}
    \\
    \,&=\,
    T_{1} T_{2} + T_{3} \big(1 - T_{2}\big)
    \tagaligneq \label{eq:consistency_decomposition}
    .
\end{align*}
The assumed condition \eqref{eq:fusion_test_validity} on the fusion test implies that $\lim_{n\to\infty} T_{2} = T_2^\infty$, where 
\begin{align*}
    T_2^\infty = \begin{cases}
        0, \qquad &\textrm{if } D(Q_c, Q_h ) > \theta, \\
        a, \qquad &\textrm{if } D(Q_c, Q_h ) = \theta, \\
        1, \qquad &\textrm{if } D(Q_c, Q_h ) < \theta,
    \end{cases}
\end{align*}
for some constant $a \in (0, \alpha_f]$. 
Moreover, since $T_1$ (resp., $T_{3}$) is the probability of rejecting the causal null hypothesis $\cH_0$ given that $\Qcm$ and $\Qhl$ are merged (resp., \emph{not} merged), the assumed consistency of the causality test, namely \eqref{eq:causal_test_consistency_no_merge} and \eqref{eq:causal_test_consistency}, implies $\lim_{n \to \infty}T_1 = \lim_{n \to \infty}T_{3} = 1$. Combining these arguments with the previous decomposition gives
\begin{align*}
    \lim\nolimits_{n \to \infty} \Pr\Big(\phi_\ast\big(\Qcm, \Qhl, \Qtn\big) = 1\Big)
    = T_2^\infty \lim\nolimits_{n \to \infty} T_1  + \big(1-T_2^\infty\big) \lim\nolimits_{n \to \infty} T_3
    = 1 .
\end{align*}
\end{proof}

\subsection{Proof of \Cref{cor:validity}}
\label{pf:cor:validity}
\begin{proof}
It suffices to verify the conditions in \Cref{thm:validity_general} hold for the proposed fusion test and the causality tests in \Cref{alg:equiv_ttp}. 

Firstly, since the fusion test is an MMD equivalence test, \Cref{thm:mmd_equiv_test} immediately shows that the condition \eqref{eq:fusion_test_validity} holds. The second condition \eqref{eq:causal_test_validity_no_merge} requires the causality test to be well-calibrated when the fusion test is \emph{not} rejected. In this case, the causality test is a standard MMD test between $\Qcm$ and $\Qtn$. Hence, \eqref{eq:causal_test_validity_no_merge} follows from existing result on the validity of the standard MMD permutation test; see, e.g., \citet[Proposition 1]{schrab2023mmd}.

It remains to verify \eqref{eq:causal_test_validity} for the causality test with either the partial bootstrap approach described in \Cref{alg:partial_bootstrap} or with the partial permutation approach in \Cref{sec:partial_permutation}. Note that \Cref{thm:partial_boots_validity} and \Cref{thm:partial_perm_validity} cannot be used directly to show this condition, since the probability in \eqref{eq:causal_test_validity} is \emph{conditional} on the event $\{\phi_f(\Qcm, \Qhl) = 1\}$, whereas \Cref{thm:partial_boots_validity} and \Cref{thm:partial_perm_validity} concern the \emph{unconditional} probabilities. Nevertheless, \Cref{thm:partial_boots_validity_conditional} and \Cref{thm:partial_perm_validity_conditional} show that \eqref{eq:causal_test_validity} still holds even after conditioning on that event, so long as the fusion test satisfies the following condition whenever $D(Q_c, Q_h) \leq \theta$, 
\begin{align*}
    \tau \coloneqq \lim\nolimits_{n \to \infty} \Pr(\phi_f(\Qcm, \Qhl) = 1) > 0 .
\end{align*}
Since the fusion test is an MMD equivalence test, \Cref{thm:mmd_equiv_test} shows that there exists some $a \in (0, \alpha_f]$ with
\begin{align*}
    \tau = \begin{cases}
        a, & D(Q_c, Q_h) = \theta , \\
        1 , & D(Q_c, Q_h) < \theta.
    \end{cases}
\end{align*}
Therefore, \eqref{eq:causal_test_validity} holds, and all conditions in \Cref{thm:validity_general} are met, thus proving \Cref{cor:validity}. 
\end{proof}

\subsection{Proof of \Cref{cor:consistency}}
\label{sec:proof_cor_consistency}
\begin{proof}
   As we conduct the MMD equivalence test for the fusion test, condition \eqref{eq:fusion_test_validity} is satisfied. 
   
   When the assumed conditions hold, under $\cH_1$, if $\phi_f\big(\Qcm,\Qhl\big)=0$, the historical control is not fused. In this case, both partial permutation and partial bootstrap revert back to the classical two sample MMD testing between $\Qcm$ and $\Qhl$. \citet{gretton2006kernel} shows that the MMD computed with characteristic kernel converges to a positive constant under a fixed alternative when two distributions are different. Consistency of the resulting permutation or bootstrap calibrated MMD test then follows from general results on permutation and bootstrap tests (\citealp[Theorem 15.1]{lehmann2022testing}; \citealp[Theorem 23.1]{van2000asymptotic}). Another way to prove the consistency is by setting $\gamma = 1$ in both \Cref{sec:proof_partial_boots_consistency} and \Cref{sec:proof_partial_perm_consistency}. Thus, \eqref{eq:causal_test_consistency_no_merge} is met.
   
   When conditioning on $\{\phi_f\big(\Qcm,\Qhl\big)=1\}$, given $D(Q_c,Q_h)\leq \theta$, we know $\tau\coloneqq \Pr\big(\phi_f\big(\Qcm,\Qhl\big)=1\big)>0$. 
   \Cref{thm:partial_boots_consistency} shows that the causality test with partial bootstrap obtains \eqref{eq:unconditional_consistency}  if and only if $Q_c=Q_h$, or $Q_c \neq Q_h$ and $2(1-\gamma)D\big(Q_h,Q_c\big)\cos \beta < D(Q_c, Q_t)$. Similarly, with partial permutation, \eqref{eq:unconditional_consistency} stands if and only if $Q_c=Q_h$, or $Q_c \neq Q_h$ and $2\cos\beta\, D\big(Q_h,Q_c\big)<\Big(\frac{1}{1-\gamma}+\lambda\Big)D\big(Q_t,Q_c\big)$.  Hence, all conditions in \Cref{prop:consistency_conditional_merged} are met, thus we obtain \eqref{eq:causal_test_consistency} for both partial permutation and partial bootstrap.
   
   Applying \Cref{thm:consistency_general} then shows the claimed result in \Cref{cor:consistency}.
\end{proof}

\subsection{Proof of \Cref{thm:partial_perm_consistency} with converging sample size ratios}
\label{sec:example_varying_ratio}
In this section we illustrate \Cref{remark:generalize_fixed_ratio}, building on
\Cref{sec:proof_partial_perm_consistency}, by showing how to adapt the argument
to the case where the sample ratios merely converge to constants, rather than
being exactly fixed. Concretely, we now assume
$\ell/n \to c_1$ and $m/n \to c_2$. Correspondingly, we define
\begin{align*}
  \gamma_n \coloneqq \frac{m}{m+\ell}, \qquad
  \lambda_n \coloneqq \frac{n}{m+n},    
\end{align*}
and denote by $\gamma$ and $\lambda$ the fixed ratios appearing in the main
text. Under the convergence of the sample ratios we then have
$\gamma_n \to \gamma$ and $\lambda_n \to \lambda$.

\begin{proof}
To adapt the proof in \Cref{sec:proof_partial_perm_consistency} to this
setting, it suffices to show, for example, that
\begin{align*}
    \tnml = \big\|(1-\gamma_n)\big(\hat{\xi}_{h,\ell}-\hat{\xi}_{c,m}\big)-\big(\hat{\xi}_{t,n}-\hat{\xi}_{c,m}\big)\big\|^2
\end{align*}
converges to 
\begin{align*}
    \big\|(1-\gamma)\big(\xi_{h}-\xi_{c}\big)-\big(\xi_t-\xi_{c}\big)\big\|^2. 
\end{align*}
One could view $D^2(\Qcm,\Qtn,\Qhl,\gamma_n)$ as a function of not only $\Qcm$, $\Qtn$, $\Qhl$, but also $\gamma_n$, and
invoke \Cref{lem:hadamard_mmd} directly. Here we instead pursue a simple
decomposition which makes the technical steps explicit and can be reused when
extending other results to converging ratios.

We rewrite $T_{n,m,\ell}$ as
\begin{align*}
  T_{n,m,\ell}
  &= \big\|(1-\gamma_n)\big(\hat{\xi}_{h,\ell}-\hat{\xi}_{c,m}\big)
         -\big(\hat{\xi}_{t,n}-\hat{\xi}_{c,m}\big)\big\|^2 \\
  &= \big\|(1-\gamma_n)\big(\hat{\xi}_{h,\ell}-\hat{\xi}_{c,m}\big)
      -(1-\gamma)\big(\hat{\xi}_{h,\ell}-\hat{\xi}_{c,m}\big) \\
  &\hspace{4em}
      + (1-\gamma)\big(\hat{\xi}_{h,\ell}-\hat{\xi}_{c,m}\big)
      -\big(\hat{\xi}_{t,n}-\hat{\xi}_{c,m}\big)\big\|^2 \\
  &= \big\|\underbrace{(\gamma-\gamma_n)
        \big(\hat{\xi}_{h,\ell}-\hat{\xi}_{c,m}\big)}_{E_n}
      + \underbrace{(1-\gamma)\big(\hat{\xi}_{h,\ell}-\hat{\xi}_{c,m}\big)
        -\big(\hat{\xi}_{t,n}-\hat{\xi}_{c,m}\big)}_{F_n}\big\|^2.
\end{align*}
The term $F_n$ is exactly the finite sample expression appearing in the
fixed-ratio case (with $\gamma$ fixed).

Therefore, with Cauchy-Schwartz, we obtain
\begin{align*}
    \|E_n\|  &=|\gamma-\gamma_n| \big\|\hat{\xi}_{h,\ell}-\hat{\xi}_{c,m}\big\|\\
    &\leq |\gamma-\gamma_n| \sup\Big(\big\|\hat{\xi}_{h,\ell}\|+\|\hat{\xi}_{c,m}\big\|\Big)\\
    &\leq 2K|\gamma - \gamma_n|,
\end{align*}
where $K$ is a finite
constant such that the kernel embeddings are bounded by it given \Cref{ass:kernel_moment}. With $|\gamma-\gamma_n|\to0$, we obtain $\|E_n\|\to 0$.

From the fixed ratio analysis in \Cref{sec:proof_partial_perm_consistency} we
know that $F_n \to (1-\gamma)(\xi_h-\xi_c) - (\xi_t-\xi_c)$ in $\cH_k$, hence
$\|F_n\|$ is almost surely bounded for all large $n$. Together with
$\|E_n\|\to 0$ this implies
\begin{align*}
   \Big|T_{n,m,\ell} - \big\|F_n\big\|^2\Big| \to 0 
\end{align*}
almost surely. Finally, since we already established in the fixed-ratio case
that
\begin{align*}
\|F_n\|^2 \longrightarrow
  \big\|(1-\gamma)\big(\xi_h-\xi_c\big)
      -\big(\xi_t-\xi_c\big)\big\|^2,    
\end{align*}
we conclude that
\begin{align*}
    T_{n,m,\ell}
  \longrightarrow
  \big\|(1-\gamma)\big(\xi_h-\xi_c\big)
      -\big(\xi_t-\xi_c\big)\big\|^2,
\end{align*}
as claimed.
\end{proof}

\section{Additional experiment results on normal approximation}
\label{sec:further_exp_normal_approx}
We provide additional experimental results to examine the normal approximation. As discussed in \Cref{sec:normal_approximation}, the linear kernel is unbounded, so its asymptotic behavior emerges only when the sample size is sufficiently large.

In \Cref{fig:test_statistic_approximation_clt}, the approximation accuracy with the normal approximation method improves with larger sample sizes.  As shown in \Cref{fig:synthetic_mean_shift_clt}, which uses the same mean shift setting as \Cref{fig:synthetic_mean_shift}, the normal approximation underperforms relative to the partial bootstrap when the sample size is small, leading to reduced statistical power.

\begin{figure}[t]
    \centering
    \includegraphics[width=1\linewidth]{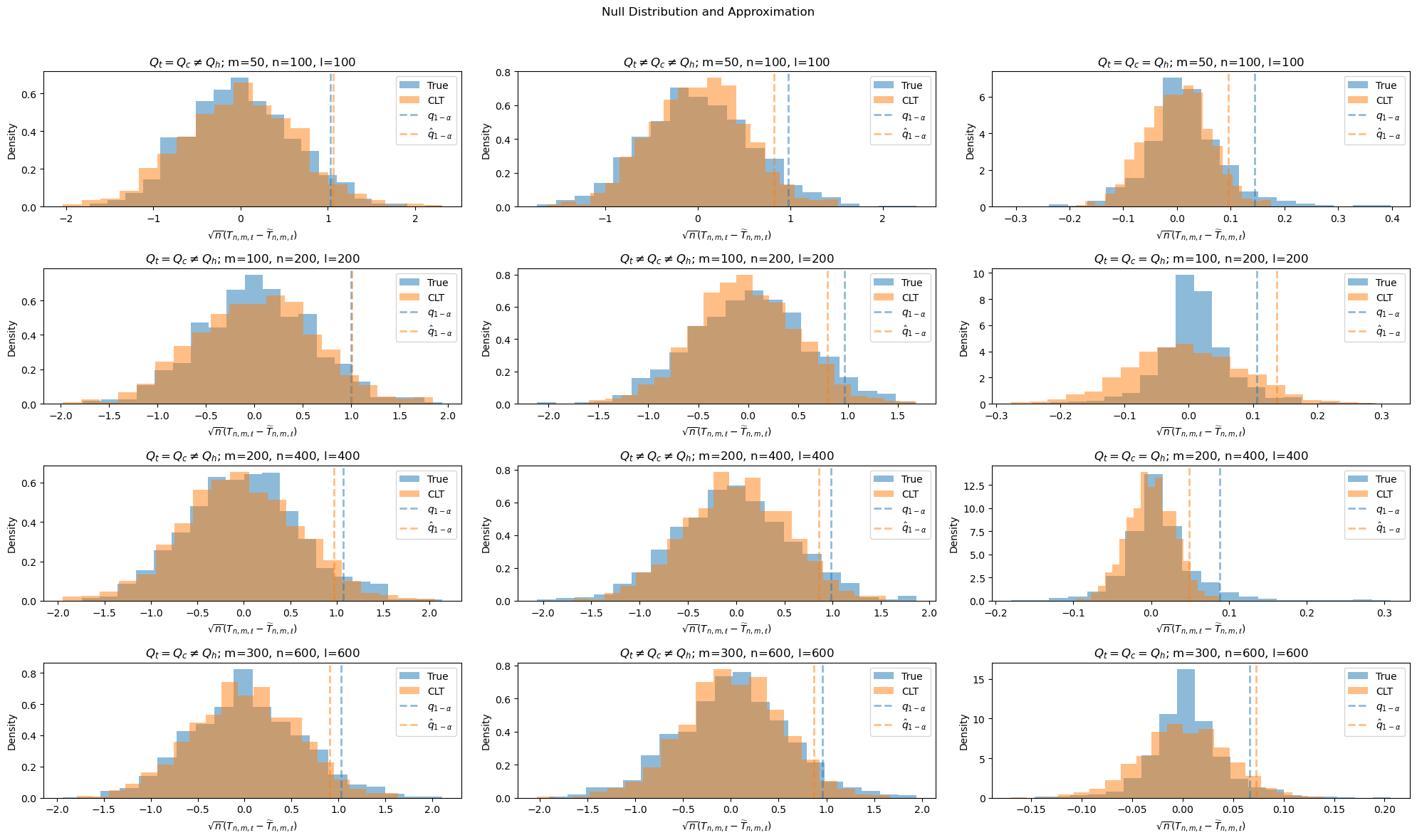}
    \caption{Comparison of the true and approximate distributions of the test statistic under partial bootstrap using normal approximation, under different sample sizes. Results from $1000$ simulations. Within each simulation, $B=1000$ bootstraps are performed.}    \label{fig:test_statistic_approximation_clt}
\end{figure}

\begin{figure}[!t]
    \centering
    \includegraphics[width=1\linewidth]{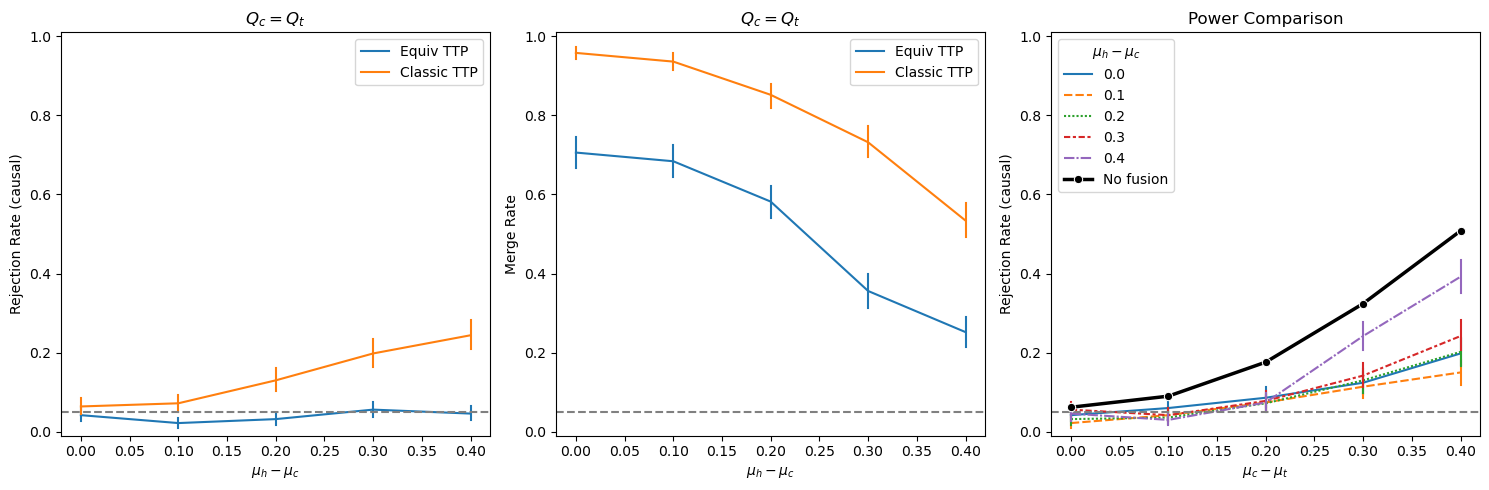}
    \caption{$\theta$ is fixed to be $0.4$, RBF kernel, mean shift, normal approximation.}
    \label{fig:synthetic_mean_shift_clt}
\end{figure}

\section{Extending Equivalence-TTP to U-statistic Estimators}
\label{app:u_stats}
We discuss how to extend our proposed TTP framework when the U-statistic \eqref{eq:mmd_u_stat} is used to estimate MMD. Throughout this work, we have relied on the V-statistic estimator \eqref{eq:mmd_v_stat} because it is guaranteed to be non-negative. This property is essential in the fusion test, namely the equivalence MMD test described in \Cref{sec:equivalence_test}. Specifically, the test statistic in the fusion test is 
\begin{align*}
    \Delta_{m,\ell}^f \,\coloneqq\, \theta - D(Q_c^m, Q_h^\ell) ,
\end{align*}
which requires computing the \emph{non-squared} MMD $D(Q_c^m, Q_h^\ell)$. With the V-statistic \eqref{eq:mmd_v_stat}, this is computed by simply taking the square root, which is always well-defined due to non-negativity. However, with a U-statistic \eqref{eq:mmd_u_stat}, the MMD estimate can take negative values, particularly when the current control $Q_c$ and the historical control $Q_h$ are similar so that $D\big(Q_h,Q_c\big) \approx 0$, or when the sample sizes are small so that \eqref{eq:mmd_u_stat} has a large variance. In those cases, the square root required to compute $\Delta_{m,\ell}^f$ is not well-defined. A simple workaround is to never reject the null hypothesis \eqref{eq:equiv_fusion_hypothesis} whenever the U-statistic is negative. However, this would prevent merging historical samples even when $D\big(Q_h,Q_c\big)$ is essentially zero, which would lead to loss of test power. 

For these practical reasons, we adopted the V-statistics estimator in the fusion test. In contrast, the \emph{causality test} can be readily extended to U-statistics, because neither the partial bootstrap nor the partial permutation procedures require taking square roots of the MMD estimator.

\bibliography{paper-ref}
\end{document}